\documentclass[11pt]{article}

\usepackage{amsthm,amsfonts,amsmath,amsfonts,amsmath,amssymb}
\usepackage{xspace}
\usepackage{algorithm,algorithmic}
\usepackage[textsize=tiny,disable]{todonotes}
\usepackage[shortlabels]{enumitem}
\usepackage{epsfig}
\usepackage{wrapfig}
\usepackage{times}
\usepackage{tikz}
\usetikzlibrary{decorations.pathmorphing}
\usetikzlibrary{calc}
\usepackage{authblk}
\usepackage{url}
\usepackage{xcolor}

\DeclareMathAlphabet{\mathbbold}{U}{bbold}{m}{n}

\newtheorem{theorem}{Theorem}
\newtheorem{claim}[theorem]{Claim}
\newtheorem{definition}{Definition}
\newtheorem{problem}{Problem}
\newtheorem{lemma}[theorem]{Lemma}

\newtheorem{corollary}[theorem]{Corollary}

\ifdefined\final
\newcommand{\comment}[1]{}
\newcommand{\dqcnote}[1]{}
\newcommand{\rrnote}[1]
\newcommand{\ignore}[1]{}

\else
\newcommand{\comment}[1]{\textsl{\small[#1]}\marginpar{\tiny\textsc{To Do!}}}
\newcommand{\rrnote}[1]{\begingroup\color{red!60!black}Ravi: #1\endgroup}
\newcommand{\ignore}[1]{}
\fi

\usepackage{hyperref}

\hypersetup{
    colorlinks=true,
    linkcolor=blue,
    filecolor=magenta,      
    urlcolor=cyan,
    citecolor=red,
    pdftitle={Overleaf Example},
    pdfpagemode=FullScreen,
    }
\begin{document}

\title{Timeliness Through Telephones:\\Approximating Information Freshness in Vector Clock Models
\thanks{This material is based upon work supported in part by the U. S. Office of Naval Research under award number N00014-21-1-2243 and the Air Force Office of Scientific Research under award number FA9550-20-1-0080.}} 
\author[1]{Da Qi Chen}
\author[2]{Lin An}
\author[2]{Aidin Niaparast}
\author[2]{R. Ravi}
\author[3]{Oleksandr Rudenko}

\affil[1]{Biocomplexity Institute and Initiatives, University of Virginia, USA. \texttt{wny7gj@virginia.edu}}
\affil[2]{Tepper School of Business, Carnegie Mellon University, USA. \texttt{{linan,aniapara,ravi}@andrew.cmu.edu}}
\affil[3]{Department of Mathematical Sciences, Carnegie Mellon University, USA. \texttt{orudenko@andrew.cmu.edu}}

\date{\today}
%\author{Authors blinded for review}
%\date{}

\maketitle

\begin{abstract}
We consider an information dissemination problem where the root node in an undirected graph constantly updates its information. The goal is to keep every other node in the graph about the root as freshly informed as possible.
Our synchronous information spreading model uses telephone calls at each time step, in which any node can communicate with at most one neighbor, thus forming a matching over which information is transmitted at each step. 
%These matchings that define the schedule over time ad infinitum are specified in a centralized manner. 
We introduce two problems in minimizing two natural objectives (Maximum and Average) of the latency of the root's information at all nodes in the network.

After deriving a simple reduction from the maximum rooted latency problem to the well-studied minimum broadcast time problem, we focus on the average rooted latency version.
%Even though our feasible solutions are infinite schedules of matchings, it is not hard to argue that it suffices to consider linear-sized periodic schedules with a loss of at most a factor of two in each of these objectives.
%Next, by relating the maximum rooted latency problem to minimizing the maximum broadcast time, we derive new approximation algorithms for this version. 
We introduce a natural problem of finding a finite schedule that minimizes the average broadcast time from a root.
%spanning tree in the given graph that minimized the average broadcast time
%A similar approach for the average rooted latency objective motivates a new average broadcast time problem in graphs that we introduce.
We show that any average rooted latency scheme induces a solution to this average broadcast problem within a constant factor and conversely, this average broadcast time is within a logarithmic factor of the average rooted latency.
Then, we derive a log-squared approximation algorithm for the average broadcast time problem via rounding a time-indexed linear programming relaxation, resulting in a log-cubed approximation for the average latency problem. 
%for the average broadcast time problem and show how to round it using techniques developed for the maximum broadcast time problem. 
%to derive poly-logarithmic approximation algorithms for average broadcast time and the average rooted latency problems.

Surprisingly, we show that using the average broadcast time for average rooted latency introduces a necessary logarithmic factor overhead even in trees. We overcome this hurdle and give a 40-approximation for trees. For this, we design an algorithm to find near-optimal locally-periodic schedules in trees where each vertex receives information from its parent in regular intervals. On the other side, we show how such well-behaved schedules approximate the optimal schedule within a constant factor. 

%Our LP-based reduction from the average broadcast time problem to the maximum broadcast time version also works for general requirement patterns of information spreading where specific sources are required to transmit to specific destinations:  this also leads to new approximation algorithms for the average latency problems involving multi-commodity demands.
%They also apply to the whole range of $\ell_p$ norms of the latency objective that span the range from maximum to average latency. 
%Thus, our results provide the first approximation algorithms for a large set of latency problems. 
%that are locally periodic at each node and where each node has the same offset value in hearing fresh information from its parent. The average latency version  motivates the problem of finding a spanning tree that minimizes a new Average Broadcast Time objective in graphs, which is interesting in its own right for future study.

\end{abstract}
\thispagestyle{empty}
\newpage
\pagenumbering{arabic}
\section{Introduction}

%\subsection{Models of Information Dissemination and Freshness}
Rumor spreading problems, where a crucial piece of information needs to be disseminated throughout a given network as quickly as possible, have been an active research topic for decades. They have applications ranging from increasing throughput in synchronizing networks~\cite{Farley1980Broadcast}, keeping object copies in distributed databases synchronized~\cite{demers1987epidemic}, scheduling inter-satellite communication networks~\cite{chu2018time}, to recreational mathematics~\cite{hajnal1972cure}.
Common minimization objectives in rumor spreading involve the total number of messages, the total number of transmissions (especially when message sizes are bounded in transmissions) and the completion time. This class of problems have a bounded time horizon of interest within which the required messages reach their requisite destinations.

In a related direction, a notion of information latency in social networks was introduced by Kossinets et al. in~\cite{kossinets08}. These problems are inspired by models of keeping clocks synchronized in distributed computing~\cite{lamport,mattern}. %The goal is to examine how normal patterns of information exchange in such networks affect the freshness of information at various nodes in the network. In these models, nodes are constantly updating information over time and the regular patterns of communication result in an information latency at each node with respect to the latest available information at all other nodes. The goal is to understand the effect of the communication patterns on the general features of this information latency in the long run. This class of problems have potentially an unbounded time horizon and the metrics focus on summary properties of information freshness over time. 
Motivated by this vector clock model, in this paper, we draw the first connection between various delay objectives in rumor spreading  problems with finite schedules and the notion of information latency in infinite time schedules.

%We first introduce the infinite-horizon vector clock model of information freshness. This model simulates a scenario where a node acquires fresh information at every time-step, such as a camera that captures new images every hour, and we seek the best way to update everyone else of newly captured information. We then use this model to define the more studied minimum delay broadcast time problem~\cite{Ravi1994Broadcast,Elkin2003SublogTel}, which corresponds to the situation when the camera only captures one image and needs to spread it to everyone else. 
%One advantage of going in this direction rather than chronologically is that we see the motivation for some new minimum delay problems from this approach.

\subsection{Rooted Vector Clock Models for Minimizing Latencies}

In this model, a root node $r$ in an undirected graph $G=(V,E)$ (where $n= |V|$) is constantly updating its information at every time-step and
every node in the graph is interested in the most updated information from the root. We assume information is exchanged 
according to the {\bf telephone} model~\cite{Hedetniemi1988GossipSurvey}, where in each round any node can communicate with at most one neighbor (modeling a telephone call).
Therefore, a set of edges corresponding to a matching is picked at each time-step and both ends of each edge in the matching exchange the
information they have about the root (and thus the more recently informed node updates the other node). A schedule is then a sequence of matchings. The objective is to minimize the latency of any node in the network about the fresh information in the root \emph{at any given time}.
%In the (single) rooted problem, the objective is to minimize the maximum latency of the root's status at any node at any point in time. 
%However, the network connectivity pattern may be a bottleneck for a node to be refreshed about another node’s status. For instance, there will a latency of at least be as much as the diameter of the network about the status of one end of the network to be monitored or recorded at the other end.

To formalize this model, we modify and use the notation developed for vector clocks~\cite{kossinets08,lamport,mattern}: 
Fix a certain schedule,
%of communications over time (in the telephone model, this is a time-ordered sequence of matchings). 
let $\phi_{v}^t$ be the {\bf view} that node $v$ has of the root $r$ at time $t$, representing the latest information $v$ has about the root. More precisely, it is the largest time $t' \leq t$ for which the information of $r$  at time $t'$ could be transmitted through the sequence of communications in our fixed schedule and arrive at $v$ by time $t$. Define $\phi_v^0 = 0$ for all nodes $v$ and $\phi_r^t = t$ for all time $t$. When two nodes $u, v$ exchange information, they update their view as follows:
%\begin{equation*}
 $   \phi_v^{t+1}=
               \max \{ \phi_v^{t}, \phi_u^{t} \} \quad \forall v\neq r.$
%\end{equation*}

The {\bf information latency} at vertex $v$ at time $t$, denoted $l_v^t$, is given by $t - \phi_v^t$. This represents the freshness of the information $v$ has about $r$ at current time $t$. An alternate interpretation of the latency is that if the current latency $l_{v}^{t} = k$, then in the most recent $k$ steps, there is a path of matched edges along which communication was scheduled in monotonic order (i.e., from $r$ to $v$) that was able to convey the information at $r$ from $k$ steps earlier to $v$ by the current time $t$.
Two natural goals of a centralized information exchange scheme is to minimize the maximum and total (or equivalently, the average) information latency of the nodes at all times. 
\begin{problem} [Max Rooted Vector Clock (MaxRVC)]
Given a graph $G$ with a root $r$, find an infinite schedule that minimizes $\max_{v\in V(G), t\ge 0} l_v^t$. 
\end{problem}

%It is convenient to think of a communication graph $G[t_1,t_2]$ for integer time indices $t_1 \leq t_2$ which is the union of subgraphs (matchings in the telephone model) over which communication was scheduled, where the edges in the subgraph are further labeled with the time $t \in [t_1,t_2]$ in which communication in the edge was scheduled. If the same edge was scheduled many times during this duration, we may introduce multiedges with different time indices to denote this. To rephrase the definition of $l_{v}^{t}(u)$, we simply find the largest $t'$ such that the graph $G[t',t]$ contains an increasing-label path from $u$ to $v$, and denote $l_{v}^{t}(u) = t-t'$ for this choice. Note that $t'=\phi_v^t(u)$ in this formulation.

%Instead of minimizing the maximum latency at any node, one can also consider the average latency across all nodes. To formally define this objective, let $l^t$ be the vector that contains the set of all latencies in the network at the current time $t$ where $l^t_v$ appears at entry $v$ in the vector. 
\begin{problem} [Average Rooted Vector Clock (AvgRVC)]
Given a graph $G$ with a root $r$, find an infinite schedule that minimizes $\max_{t\ge 0} \frac{1}{n} \sum_{v\in V(G)} l_v^t$. 
\end{problem}

\subsection{Broadcast Time Problems}
We relate the latency objectives defined above for infinite schedules to a `static' one-phase version of delay that has been widely studied in rumor spreading problems.
Here, the notion of latency is replaced by the notion of delay, the first time at which the information of interest is received.
%newly cut+++++++++++++++++++++++++++++++
Consider the one-shot analog of MaxRVC, the classic Minimum Broadcast Time Problem.
We can cast this problem using our latency notation as follows. The goal is to pick a {\bf finite} vector-clock schedule where all nodes (other than the root) start with $l_v^0 = -\infty$. Define the first time at which the latency of any node becomes finite as the {\bf delay} in the root's information (at time 0) reaching this node. The minimum broadcast time problem aims to minimize the maximum delay of the root's information (at time 0) reaching any node.
\begin{problem} [Minimum Broadcast-Time Problem]
Given a graph $G$ and a root $r$, find a finite schedule (under the telephone model) that minimizes the maximum delay of any node $v$. 
\end{problem}

%\subsection{Related Work on Minimum Broadcast Time and Variants}
%\subsubsection{Telephone model}
Building on prior work by Korsartz and Peleg~\cite{KP}, Ravi~\cite{Ravi1994Broadcast} designed the first poly-logarithmic approximation algorithm for the Minimum Broadcast Time problem in an arbitrary network.
This work related such schemes to finding spanning trees that simultaneously have small maximum degree and diameter (the so-called {\it poise} of the graph).
After some follow-up work of Bar-Noy et al.~\cite{BGNS98}, Elkin and Kortsarz~\cite{Elkin2003SublogTel} gave the currently best-known $O(\frac{\log k}{\log \log k})$-approximation factor, where $k$ is the number of terminals the information should be delivered to in the more general multicast version.
On the hardness side, Middendorf ~\cite{middendorf} showed that the minimum broadcast-time problem is NP-hard even on 3-regular planar graphs and Elkin and Kortsarz \cite{ek-05} proved that it is hard to approximate better than a factor of 3 unless P$=$NP.

Similar to AvgRVC, we consider the average-time analogue (ABT) of the Minimum Broadcast-Time Problem. 
To the best of our knowledge, ABT was not studied before.
We define the ABT problem and design the first approximations for ABT and 
use it to obtain approximations for AvgRVC. 
\begin{problem} [Average Broadcast-Time Problem (ABT)]
Given a graph $G$ and a root $r$, find a finite schedule (under the telephone model) that minimizes $\frac{1}{n} \sum_{v\in V(G)} d_v$ where $d_v$ is the delay at $v$. 
\end{problem}
Note that if $G$ is a rooted tree, the optimal solution to ABT corresponds to ordering the children at any internal node in non-increasing order of the sizes of the subtrees under them, and sending the information from the root forth in that order greedily.

\subsection{Our Contributions}
%Note that the minimum broadcast-time problem under the telephone model is NP-hard even on 3-regular planar graphs~\cite{middendorf} and hard to approximate to better than a factor of 3~\cite{ek-05}; 
%Consequently, we expect MaxRVC and AvgRVC problems to be NP-hard.
%the same hardness results carry over to MaxRVC.
%In our new models of information freshness, w
%Below is an extended summary of our results.
%We summarize our results and techniques used below. Our two main technical results are points 4 and 6 below and are further elaborated in Section~\ref{sec:techniques}.
%\begin{enumerate}
%\item \textbf{Existence of near-optimal periodic schedules}: 

For any $\ell_p$-norm of the latencies, we show that there exists a periodic schedule that has period at most $opt+1$ and is 2-optimal with respect to the optimal infinite schedule; Here, $opt$ is the minimum value of the corresponding norm of the latencies and is at most $n$, the number of nodes in the graph (Theorem \ref{periopt}). 
%Similar result also holds for time-average objectives (Theorem \ref{periopt:time})
%The main idea for the proof is to repeat the first $opt$ steps of an optimal schedule and examine its latencies. 
Note that this gives us a huge advantage for solving vector clock problems since it works with many objective functions and reduces the search of an infinite schedule to only finite ones.

%\item \textbf{MaxRVC and Minimum Broadcast-Time}: 
Given an $\alpha$-approximation for the Minimum Broadcast-Time Problem, repeating it ad infinitum gives a $2\alpha$-approximation for MaxRVC (Theorem \ref{redn} in Section \ref{sec:prelim}). Combined with previously known approximations for the Minimum Broadcast-Time Problem, we obtain a $2$-approximation for MaxRVC on trees and an $O(\frac{\log n}{\log \log n})$-approximation for MaxRVC on general graphs.
%These results naturally extend to the multicast and the multicommodity variants of these problems.
%We prove these results by reducing the solutions of Minimum Broadcast Time and MaxRVC to each other with a slight loss. 
%The main idea is to repeat the optimal Minimum Broadcast Time schedule ad infinitum for the MaxRVC problem.
Furthermore, since from any MaxRVC solution, we can reconstruct a minimum broadcast time solution of the same objective value,
by combining the $2$-approximation loss with the $(3-\epsilon)$-inapproximability of Minimum Broadcast Time~\cite{ek-05}, we get a $(\frac32 - \epsilon)$-inapproximability for MaxRVC unless $P=NP$. Our main contributions are as follows.

\begin{enumerate}
\item \textbf{Relating Average Broadcast Time and AvgRVC}:
We show that the optimal values of ABT and AvgRVC are within a logarithmic factor of each other.
\begin{theorem}
    \label{th:sandwich}
    Let $G$ be an $n$-node graph with root $r$. Given a schedule for ABT with average broadcast time $T_{BC}$, we can derive (in polynomial time) an infinite schedule for AvgRVC with average latency of at most $2(\log n)T_{BC}$ at any time. Conversely, given an infinite schedule for AvgRVC with average latency of at most $T_{VC}$ at any time, we can derive (in polynomial time) a finite schedule for ABT with average delay of at most $4T_{VC}$. Thus, given a polynomial-time $\alpha$-approximation algorithm for the ABT, we can devise a polynomial-time $8\alpha \log(n)$-approximation algorithm to AvgRVC.
\end{theorem}

%We show that an $\alpha$-approximation for ABT leads to an $O(\alpha \log n)$-approximation for AvgRVC (Theorem \ref{avgdelaytolat} in Section \ref{sec:AvgRVC}) for general $n$-node graphs.
It is not hard to extract an ABT schedule by examining the portion of an AvgRVC scheme that leads to the set of latencies at any given time step.
In the other direction, to derive a solution to AvgRVC from a solution to ABT, unlike the MaxRVC analog, we need to repeat communications more often to those with lower delay than others. We show that this can be accomplished with only a logarithmic factor blow-up in the average latency (Lemma~\ref{createinf}). 
%As with the case of the maximum version, these results also extend to the multicast and the multicommodity variants of these problems.

\item \textbf{Approximating ABT and AvgRVC in general graphs}: The following is one of our main results.

\begin{theorem}
    \label{th:abtapprox}
    ABT has an $O(\frac{\log^2 n}{\log \log n})$-approximation algorithm in $n$-node graphs.
\end{theorem}

As a simple corollary of Theorems~\ref{th:sandwich} and \ref{th:abtapprox}, we can approximate AvgRVC in general graphs. 

\begin{theorem}
    \label{th:AvgRVCapprox}
    AvgRVC has an $O(\frac{\log^3 n}{\log \log n})$-approximation algorithm in $n$-node graphs.
\end{theorem}

%We derive an $O(\frac{\log^2 n}{\log \log n})$-approximation for the average broadcast time problem in an $n$-node graph in Section~\ref{sec:abt} using an LP rounding algorithm. 
%This is one of the main technical contributions of our work and we elaborate on the key ideas in Section~\ref{sec:techniques}.
We prove Theorem~\ref{th:abtapprox} by using a natural time-indexed linear programming formulation. The formulation helps us identify the power-of-two time scale in which the optimal (fractional) transmission scheme reaches each node. We use this to partition the terminals based on when they are informed. For each time scale, we project the fractional solution certifying this time scale into the original graph to get a feasible LP solution to the \emph{minimum poise Steiner tree} problem on this subset of nodes of value equal to the time scale. 
This is the problem of finding a Steiner tree connecting a given set of terminal nodes to the root with the minimum poise, defined as the sum of the maximum degree of any node in the tree and the diameter of the tree.
We then round this LP fractional solution using the results of~\cite{latin18} and~\cite{Ravi1994Broadcast} to schedule these nodes within a log-squared factor of the time scale value. By scheduling the partition of the node sets to be informed by these multicast trees in increasing order of their time scales, we ensure that there is only a constant factor loss in the approximation ratio for ABT on top of the LP based rounding.  Using the reduction mentioned above, this also gives a poly-logarithmic approximation algorithm for the AvgRVC in general graphs.

\item \textbf{AvgRVC in trees}: Unfortunately, unlike the relationship between MaxRVC and Minimum Broadcast-Time, we show that there exists an $\Omega(\log n)$-gap between an optimal value for AvgRVC and for ABT even for the following tree $T$. Let $r$ be the root of $T$. Let $v_1, ..., v_d$ be the children of $r$. Suppose each child $v_i$ is the root of a complete binary tree of size $n_i=\lceil n/(2i^2)\rceil$. Choose $d$ such that $\sum_{i=1}^d n_i = n$. Note that $d\ge \sqrt{n}$.
%To upperbound the ABT value, one can consider the schedule where the root informs $v_1, ..., v_d$ in order and every other vertex informs its left child before its right child as soon as they hear the message. With some calculation, one can show that this schedule yields an average broadcast time of at most $O(\log n)$.
%Lower bounding the AvgRVC value is more complicated since the optimal schedule may behave very erratically. It turns out a lower bound can be achieved by ``regularizing" the optimal schedule, yielding a bound of $\Omega(\log^2n)$. 
The details are in Section \ref{sec:gap}. 
%The methods here also helpfully point to a strategy to explicitly construct a near optimal schedule on trees.

Since the optimal value for ABT has an $\Omega(\log n)$-factor gap with that of AvgRVC even for trees, we further study the structure of optimal vector clock schedule on trees to overcome this gap.
\begin{theorem}
\label{treeapprox}
There exists a $40$-approximation for the Average Rooted Vector Clock Latency Problem (AvgRVC) on trees. 
\end{theorem}
We obtain this result by devising near-optimal schedules that are locally periodic at each node and where each node has the same wait-time (offset value) in hearing fresh information from its parent
%, we design a 40-approximation for AvgRVC on trees (Theorem \ref{treeapprox}). 
This is our most technical contribution (Section~\ref{sec:tree}).
%, so we elaborate on the techniques used in Section~\ref{sec:techniques}.
Since these techniques directly tackles the vector clock problems, they potentially provide alternate strategies in designing algorithms for infinite schedules in general without reducing them to delay problems.

\end{enumerate}

\subsection{Related Work}
\label{rel-work}
\paragraph{Minimum Broadcast Time.} The most widely studied problem related to our work is that of minimizing the maximum broadcast time under the telephone model. 
In addition to being hard to approximate to within a factor better than 3 in general graphs unless P$=$NP~\cite{ek-05}, it is known to be NP-complete even on special classes of graphs such as bipartite
planar graphs, grid graphs, complete grid graphs, split graphs and chordal graphs~\cite{jansen1995minimum}.
Several heuristic approaches have been proposed for the problem \cite{scheuermann1984heuristic,hasson2004novel,harutyunyan2006efficient,harutyunyan2014new,de2018heuristics}, as well as approaches based on integer programming~\cite{ivanova2021computing}.
%first, a compact integer programming model is developed. While the model targets the exact minimum in instances of moderate size, its continuous relaxation is suitable for computation of lower bounds in larger instances. Second, we derive lower bounding techniques, both of an analytical nature and in terms of a combinatorial relaxation of MBT, that do not rely on linear programming. Third, we devise an upper bounding algorithm, which in combination with the lower bounds is able to close the optimality gap in a wide range of instances.
%Applications of the minimum broadcast time problem arise in keeping
%\paragraph{Applications} \cite{chu2018time}   Inter-satellite communication when one or a few satellites need to broadcast data quickly by means of time-division multiplexing.
%\paragraph{Variations of the basic telephone model.}
%In $c$-broadcasting, which is a generalization of regular broadcasting, up to $c$ adjacent nodes are permitted to receive the message from a single node per time unit. 

In~\cite{latin18}, Iglesias et al. designed an $O(\frac{\log^3k\log n}{\log\log n})$-approximation for a multi-commodity multicast generalization for $n$-node planar graphs with $k$ terminals.
One of the key contributions of their work is to show how to round a natural linear programming relaxation for minimum poise Steiner trees within a logarithmic factor.

\paragraph{Packet Routing.} In the store-and-forward packet routing model, the objective is to minimize source-to-destination packet delay, but the constraint is on the edges rather than the nodes, i.e., at most one packet crosses an edge at each round. However, unlike our models, the transmission paths are typically specified in this problem leaving only the scheduling component. A simultaneous multicast (unicast) instance in this model is given by a set of rooted trees (paths) on a network graph, and the objective is for each root to send its message to the leaves of its corresponding tree (path) through the edges of the tree (path) in the minimum number of rounds. These problems are studied in \cite{leighton1994packet, ghaffari2015near, haeupler2019near}. 
%(More references in \cite{haeupler2019near}).

\paragraph{Random rumor spreading and distributed computing models.} Inspired by the application keeping distributed databases synchronized~\cite{demers1987epidemic} using a simple epidemic protocol, Feige et al.~\cite{feige1990randomized} provided an early analysis in a long line of work on what is known as uniform broadcasting: in a given graph representing the possible communication network, each informed node picks a random neighbor and communicates its information to it.
This can be viewed as a model of pushing information to a random neighbor, leading to similarly defined models of every uninformed node pulling from a random neighbor, and also the push-pull model with both types of communications at each time step.
Culminating a long line of work, Giakkoupis~\cite{giakkoupis2011tight} showed that for any graph with conductance $\phi$, the Push-Pull algorithm distributes a rumor to all nodes of the graph in $O(\phi^{-1}\log(n))$ with high probability, and that this bound is tight.
Chierichetti et al. \cite{chierichetti2018rumor} show if the degrees of the two endpoints of each edge in the network differ by at most a constant factor, then both Push and Pull also attain $O(\phi^{-1}\log(n))$ with high probability.
%Daum et al. \cite{daum2020rumor} study a push policy with bounded in-degree, uniformly at random and adversarially, while Daknama et al. \cite{daknama2021robustness} analyze the robustness of Push-Pull, Push, and Pull methods.
%\cite{avin2018breaking}: rumor spreading problem where nodes can call a partner once they learn its address (e.g., its IP address). Push and Pull attain $O(\sqrt{\log(n)})$ with high probability.

Another rumor spreading problem with random local phone calls but which considers continuous updates of rumors at each node like our model is the influential work of Karp et al.~\cite{karp2000randomized}. 
The underlying graph of possible communications is complete in this model and they show how the number of transmissions for all nodes to have reliably heard the rumor can be reduced from the $O(n \log n)$ bound for the random Push and random Pull algorithms to $O(n \log \log n)$ for a new Push-Pull algorithm.
Note that all of these simple algorithms are purely local and have no global knowledge of the network.

%\paragraph{Models from Distributed Computing.}
%The distributed computing literature has an extensive set of results on the spread of information from a source from limited communication models. 
In the extensive distributed computing literature, the two most popular models of communication are GOSSIP and LOCAL which correspond to each node communicating with all neighbors versus a single neighbor at each time step respectively. E.g., Censor-Hillel et al.~\cite{censor2017rumor} consider the rumor spreading problem, where each node has some initial input and is required to collect the information of all other nodes, and give an algorithm that takes $O(D+poly(\log(n)))$ rounds in a graph on $n$ nodes and diameter $D$. 
Unlike the above models, the algorithms here are neither local nor necessarily random.
However, as with the simple rumor spreading models, the key difference with our work is that they assume that each node initially only knows the identity of its neighbors.
In these models from distributed computing, nodes gets to know a wider local neighborhood in the underlying communication graph by repeatedly communicating with its neighbors. In contrast, we assume a central scheduler that has access to the whole communication network at the start.

\paragraph{Other rumor spreading models.} There are several alternate models suggested for spreading rumors in networks.
%but unlike our models and more similar to the distributed models, they do not assume that a central scheduler has the view of the entire graph.
Feige et al.~\cite{feige2017contagious} consider a process in random graphs where a vertex is active either if it belongs to a set of initially activated vertices or if at some point it has at least $r$ active neighbors. They define an contagious set as a set whose activation results with the entire graph becoming active, and prove sharp bounds on the size of such sets in random graphs. 
However, on worst-case instances, approximating the minimal size of a contagious set within a ratio better than $O(2^{\log^{1-\delta} n})$ where $n$ is the number of vertices is intractable for every $\delta \in (0,1)$, unless $NP \subseteq DTIME(n \cdot poly(\log n))$~\cite{chen2009approximability}.
This model itself is a specialization of a very popular stochastic model of influence spread called the independent cascades model in networks due to Kempe et al.~\cite{kempe2003maximizing}.
Another closely related model that involved both additive and multiplicative thresholds at nodes that control the spread of influence is studied in~\cite{zehmakan2021spread}.
While these models analyze the optimal placement of source nodes of the rumor and their effects, our work studies centralized schedules for rumor propagation with low latency.

\section{Relating Vector Clocks and Broadcast Problems}
%Preliminary Results}
\label{sec:prelim}
\label{sec:AvgRVC}

First, we provide two simple results to develop intuition behind good vector clock schedules. 
Theorem ~\ref{periopt} shows that although a vector clock solution is an infinite schedule, one can transform it into a periodic one with a slight loss in the objective value. Thus, for the sake of approximation, we only need to consider finite schedules. 
%Theorem ~\ref{redn} demonstrates how one obtains good approximations for MaxRVC. This motivates us to study the average objective in the later sections where the problem becomes significantly more challenging. 

%\subsection{Periodic Near-Optimal Schemes}
%\label{sec:nearopt}

%This subsection focuses on showing that any non-periodic scheme that minimizes latency under any $\ell_p$ norm can be transformed into a periodic scheme with near optimal performance. 

\begin{theorem}
	Suppose $\mathcal{S}$ is a scheme such that $||l^t||\le opt$ for all $t\in \mathbb{N}$ where $||x||$ denotes a fixed $\ell_p$-norm. Then, there exists a periodic scheme $\mathcal{S}'$ with period $opt+1$ and latency vector $l^t{}'$ such that $||l^t{}'||$ is at most $2opt$ for all $t\in\mathbb{N}$. 
	\label{periopt}
\end{theorem}

Next, we provide a simple reduction from MaxRVC to the Minimum Broadcast-Time Problem. 
The proof requires a closer examination of the schedules for the two problems and attempts to construct one from the other. 
Refer to Section \ref{sec:app} for details.
\begin{theorem}
    If there exists an $\alpha$-approximation for the Minimum Broadcast-Time Problem, then there exists a $2 \alpha$-approximation for MaxRVC.
    \label{redn}
\end{theorem}

Since the Minimum Broadcast-Time Problem is solvable exactly on trees and there exists $O(\frac{\log n}{\log \log n})$-approximations for general graphs~\cite{ek06}, we easily obtain the following 
%by combining Theorems~\ref{periopt} and
corollary to Theorem~\ref{redn}.
\begin{corollary}
For MaxRVC, there exist a $2$-approximation for MaxRVC on Trees and an $O(\frac{\log n}{\log \log n})$-approximation on general graphs. 
\end{corollary} 

%In Section \ref{sec:gen}, we discuss extensions to generalizations of the Vector Clock problem. 

%\section{Reducing AvgRVC to ABT}

%Next, we relate AvgRVC to ABT. Unlike the previous case where the objective is to minimize the  maximum instead of the average, we can no longer guarantee a constant blowup between the two objective values.

\iffalse
\begin{theorem}\label{avgdelaytolat}
	Given a polynomial-time $\alpha$-approximation algorithm for the rooted minimum average broadcast problem (ABT), we can devise a polynomial-time $8\alpha \log(n)$-approximation algorithm to AvgRVC.
\end{theorem}
\fi
%With the above-mentioned tree example, we can show that this logarithmic factor is unavoidable if we use ABT to approximate AvgRVC, even on certain trees. 

Next, we prove Theorem~\ref{th:sandwich}. 
%To prove Theorem \ref{avgdelaytolat}, we first use a vector clock schedule to find a schedule with small average broadcast time. In the other direction, we show below how to turn a finite schedule with bounded average delay into an infinite schedule with average latency that is a logarithmic factor larger. We do this by interleaving increasingly larger prefixes of the finite schedule (in powers of two) at an appropriate rate.
We begin by proving the following useful lemma that aids in creating an infinite schedule from a finite one. The idea is to interleave increasingly larger prefixes of the finite schedule (in powers of two) at an appropriate rate to control the latency of each node at any time.

\begin{lemma}
	Given a finite ordered set $S= \{S_1, S_2, ..., S_t\}$, the elements of $S$ can be placed into  an infinite sequence $\mathcal{R}$ such that for all $1\le i\le t$, any contiguous subsequence of $\mathcal{R}$ of length $2i\log t$ contains the sets $S_1, ..., S_i$, appearing in increasing order.  
	\label{createinf}
\end{lemma}

\begin{proof}
	Let $\mathcal{R}=\{R^j\}_{j\ge0}$ where $R^j$ are themselves subsequences of length $\log t$ that we will call blocks.  We put $S_i$ into block $R^j$ if $i\equiv j (\mod 2^s)$ where $2^s$ is the largest power of $2$ less than or equal to $i$. Within each block $R^j$, the elements are placed in increasing order based on their index $i$ in $S$. Formally, $\mathcal{R}=R^0_1R^0_2...R^0_{\log t}R^1_1R^1_2...R^1_{\log t}...$ and $R^j_k=S_{2^{k-1}+j'}$ where $j'$ is the remainder of $j$ after dividing $2^{k-1}$.
	Note that if $t$ is not a power of $2$, some blocks $R^j$ will be slightly shorter than others but all have length at most $\log t$. 
	
	Another way to view this sequence $\mathcal{R}$ is the following. Element $S_1$ is the first element in every block $R^i$. Then, $S_2, S_3$ alternates as the second element of the blocks. Then, $S_4, S_5, S_6, S_7$ rotates as the third element in every block, and so on. One can easily check from the definition that no two elements in $S$ is assigned to the same position $R^j_i$ in the infinite sequence $\mathcal{R}$. 
	
	Fix some $i$ where $1\le i\le t$ and $2^s+p=i$ such that $0\le p<2^s$. Fix time $t'\ge 0 $ and consider a contiguous subsequence of length $4i\log t$ starting at time $t'$ in $\mathcal{R}$. Suppose $t'$ corresponds to a time in block $R^J$. Let $I$ be the smallest multiple of $2^s$ larger than $J$. Note that $I-J < 2^s\le i$. Note that $S_1$ is in $R^I$. Then $S_2, S_3$ will appear as the second element in $R^{I+2}, R^{I+3}$ respectively. Elements $S_4, S_5, S_6, S_7$ will appear as the third element in blocks $R^{I+4}, R^{I+5}, R^{I+6}, R^{I+7}$ respectively. This pattern continues until element $S_i$, appearing in block $R^{I+i}$, ensuring the subsequence $S_1, ..., S_i$ appears in order after visiting $i$ blocks starting at $I$. 
	%$I-J+i$ subsequences of $R^j$'s. 
	Since $I+i - J \le i+2^s \le 2i$ and each block $R^j$ has length at most $\log t$, our result follows immediately. 

\end{proof}

\begin{proof}[Proof of Theorem~\ref{th:sandwich}]
Let schedule $\mathcal{S}$ be a solution to ABT rooted at $r$ with average broadcast time $T_{BC}$. Note that we may assume $\mathcal{S}$ does not stay idle and informs a new vertex every step. Thus, we may assume $\mathcal{S}$ takes at most $n$ steps to complete. We now apply Lemma~\ref{createinf} to the steps (matchings) of schedule $\mathcal{S}$ to achieve an infinite schedule $\mathcal{R}$. It follows from Lemma~\ref{createinf} that if the delay of vertex $v$ at schedule $\mathcal{S}$ is $d_v$, it gets informed about $r$ in schedule $\mathcal{R}$ in every $2d_v \log(n)$ steps. This concludes that the average latency of schedule $\mathcal{R}$ is at most $2\log(n)T_{BC}$ at any time.

For the other direction, let $\mathcal{R}'$ be an infinite schedule for AvgRVC with average latency of at most $T_{VC}$ at any time. We may assume $\mathcal{R}'$ informs each vertex at some point, because otherwise $T_{VC}$ would be infinite and we do not have anything to prove. 
Let time $t_0$ be the first time when every vertex receives some information from $r$. For any vertex $v$, in the past $l_v^{t_0}$ steps, $\mathcal{R}'$ has sent a message from $r$ to $v$. Let $\mathcal{S}'_i$ be the subschedule of $\mathcal{R}'$ from time $\max(1,t_0-2^i+1)$ to time $t_0$ for $0\le i\le I = \lceil \log t_0\rceil$. We claim that $\mathcal{S}'=\mathcal{S}'_0\mathcal{S}'_1...\mathcal{S}'_I$ provides the desired schedule.
Given a vertex $v$, suppose $2^i$ is the smallest power of $2$ such that $l_v^{t_0} \leq 2^i$. Then, in subschedule $\mathcal{S}'_i$, $r$ broadcasts an update to $v$. Since the length of the subschedules are powers of $2$, it follows that $\mathcal{S}'_i$ is completed in $\mathcal{S}'$ before time $2\cdot2^i$. Since $2^i < 2l_v^{t_0}$, it follows that in schedule $\mathcal{S}'$, vertex $v$ gets informed in at most $4l_v^{t_0}$ steps. Therefore the average broadcast time of $\mathcal{S}'$ is at most $\frac{1}{n}\sum_{v\in V(G)} 4l^{t_0}_v \le 4T_{VC}$.
%Then, consider running $\mathcal{R}_{VC}$ backwards starting at time $t_0$. Note that this schedule ensures that $r$ receives an information from vertex $v$ by time $l^{t_0}_v$. By applying Lemma \ref{lem:reverse} to this schedule, we obtain a broadcasting schedule which ensures that every vertex $v$ receives an update from the root by time $4l^{t_0}_v$. Then it follows that this broadcasting schedule has an average broadcast time of at most $\frac{1}{n}\sum_{v\in V(G)} 4l^{t_0}_v \le 4T_{VC}$. Since $T_{BC}$ is the optimal average broadcasting time, our claim follows immediately.
%It follows from Lemma \ref{createinf} that we also have an infinite schedule $\mathcal{R}$ with average latency of at most $2(\log n) T_{BC}$ at any time.

\end{proof}

\section{Average Broadcast Time}
\label{sec:abt}

%Theorem~\ref{th:abvapprox} is a simple corollary of Theorem~\ref{th:sandwich} and the following.
In this section we prove Theorem~\ref{th:abtapprox}.
Our strategy is as follows: we first use an LP on a time-indexed auxiliary graph to help us identify which set of nodes we should inform first, and partition the nodes in this order. Then, for each group in the partition, we project the fractional solution to an LP solution that finds a subgraph spanning the group with minimal poise (the sum of the diameter and the maximum degree). Lastly, we use the following two lemmas from \cite{latin18, Ravi1994Broadcast} to round this fractional poise solution to obtain a feasible schedule for each group. 

\begin{lemma}\label{lem:latin18}
\cite{latin18}
    Given a fractional feasible solution of value $L$ to POISELP (defined below), a natural LP relaxation of a minimum poise of a tree connecting a root $r$ to terminals $R$, there is a polynomial time algorithm to construct a tree spanning $r\cup R$ with poise $O(L\log k)$ where $k=|R|$. 
\end{lemma}

\begin{lemma}\cite{Ravi1994Broadcast}
    \label{lem:Ravi1994Broadcast}
    Given a tree with $n$ nodes and poise $L$, there is a polynomial time algorithm to construct a broadcast scheme of length $O(L\frac{\log n}{\log\log n})$ from any root. 
\end{lemma}

Given an instance of ABT on graph $G$ with root $r$, consider the following time-indexed graph $G'$. For every vertex $v\in G$, we make $n+1$ copies $v^0, v^1, ..., v^n$. The copy $v^i$ represents the vertex at time $i$. For all $0\le i<n$, add an arc $v^iu^{i+1}$ if and only if $vu\in E(G)$, representing that $v$ is capable of sending a message at time $i$ to $u$. We denote these newly added arcs as $E_1$. Furthermore, for all $0\le i < n$, add an arc $v^iv^{i+1}$, representing that $v$ can also hold the information to the next time step. For every vertex $v$, we add a sink $v^s$ and add arcs $v^iv^s$ for every $0\le i \le n$. The goal is to send flows from the root to all of the sinks $v^s$, representing the path of broadcasting. 

In the following LP, the variable $z_e$ represents whether edge $e\in E(G')$ is used for the purpose of communication. Constraint \eqref{abtlpdeg}, imposed only on edges in $E_1$, ensures that a vertex can either receive or send a message to only one of its neighbors at each round.
Without loss of generality, we can assume that the $z$-value of arcs of the type $v^iv^{i+1}$ that are not in $E_1$ are 1 representing that any information at $v$ can always be carried over to the next time. 
%\rrnote{Note telephone constraint makes ABT optimal solution factor of 2 longer.}
For every vertex $v\in V(G)$, we have flow variable $f_v(e)$ defined on all edges $e\in E(G')$. Constraints \eqref{abtlpflow1} - \eqref{abtlpflow3} are flow constraints, ensuring the conservation of flows and one unit of flow (hence one message) must be sent from the root $r$ and reaches its recipient $v^s$ for all $v\in V(G)$. Constraint \eqref{abtlpcons1} relates the two types of variables ensuring that flows can be sent along $e\in E(G')$ only if the edge is available. The last constraints are the non-negativity constraints. Notice in an integral solution, a flow reaching $v^s$ via vertex $v^i$ corresponds to a valid path of communication of length $i$. Thus, our objective correctly captures the sum of delays whose minimization is equivalent to solving ABT. Note that since the root can inform everyone in $n$ steps by greedily informing one additional vertex at a time, the above LP is feasible.

\begin{align}
	\min \qquad & \sum_{v\in V(G)} \sum_{e= v^iv^s, 0\le i\le n}if_v(e)\notag & [ABTLP]  \\
    %\text{s.t.}&\sum_{e \in \delta(v^i)\cap E_1}z_e \le 1; & \forall v\in V(G), 0\le i\le n  \label{abtlpdeg} \\
    \text{s.t.}&\sum_{e \in \delta^+(v^i)\cap E_1}z_e + \sum_{e \in \delta^-(v^{i+1})\cap E_1}z_e \le 1; & \forall v\in V(G), 0\le i\le n-1  \label{abtlpdeg} \\
	& \sum_{e\in \delta^-(w^i)} f_v(e) = \sum_{e\in \delta^+(w^i)}  f_v(e)  & \forall v, w\in V(G), 1\le i\le n;  \label{abtlpflow1} \\
	& \sum_{e\in \delta^+(r^0)} f_v(e) = 1 & \forall v\in V(G);  \label{abtlpflow2} \\
	& \sum_{e\in \delta^-(v^s)} f_v(e) = 1 & \forall v\in V(G);  \label{abtlpflow3} \\
	& f_v(e) \le z_e & \forall e\in E(G'), v\in V(G);  \label{abtlpcons1} \\
    & z_e, f_v(e) \geq 0  & \forall e \in E(G'), v\in V(G) \label{abtlpcons2}
\end{align}

\paragraph{Partitioning the vertices}
Given a fractional solution to ABTLP with objective value $opt_{ABTLP}$, we shall partition the vertices and round the solution into separate groups based on their fractional broadcasting time. For $v\in V(G)$, let $l(v) = \sum_{i=0}^n if_v(v^iv^s)$. Let $V_j := \{v\in V(G): 2^{j-1} \le l(v) < 2^j\}$ for $1\le j\le \lceil \log n\rceil$. The set $V_j$ essentially represents the set of vertices we expect to be informed between time $2^{j-1}$ and $2^j$. In order to achieve a schedule that uses this much time, we consider the following POISELP, introduced in \cite{latin18}. Its goal is, given a graph $G$ with root $r$ and a set of terminals $R$, to find a subgraph that connects the vertices in $\{r\} \cup R$ and minimizes the sum of the maximum degree and the maximum depth; Minimal such subgraphs are Steiner trees connecting $r$ and $R$. Let $\mathcal{P}(t, r)$ denote the set of all paths from $t$ to $r$. Given a path $P$, let $len(P)$ denote the length of path $P$. 

\begin{align}
	\min \qquad & L = L_1 +L_2 & [POISELP] \\
    \text{s.t.}&\sum_{e \in \delta(v)} x_e\le L_1 & \forall v\in V(G)\setminus r \label{poiselpdeg}\\
	& \sum_{P \in \mathcal{P}(t, r)}  y_t(P) = 1  & \forall t\in R\label{poiselppath}\\
	& \sum_{P \in \mathcal{P}(t, r)} len(P) y_t(P) \le L_2  & \forall t\in R\label{poiselpdiam}\\
	& \sum_{P\in \mathcal{P}(t, r): e\in P} y_t(P)\le x_e & \forall t\in R, e\in E(G); \label{poiselpcons1}\\
    & x_e, y_t(P) \ge 0  & \forall e \in E(G), t\in R, P\in \mathcal{P} \label{poiselpcons2}
\end{align}

In this LP, $x_e$ is an indicator variable representing whether the edge $e$ will appear in the final graph. Thus, Constraint \eqref{poiselpdeg} ensures the final graph has bounded degree. The indicator variables $y_t(P)$ is defined for every terminal $t$ and every path $P\in\mathcal{P}(t, r)$, representing whether the path is present in the final tree. Constraint \eqref{poiselppath} ensures the terminals are connected to the root and Constraint \eqref{poiselpdiam} bounds the length of these paths. Constraint \eqref{poiselpcons1} relates the two variables ensuring that a path can be used only if all its edges are present. The last constraint is the non-negativity constraint. The objective function is thus the sum of the maximum degree and maximum path length. 

\begin{lemma}
    \label{lem:MaxToPoise}
    Given graph $G$ with root $r$, an optimal solution to ABTLP and a partition $\{V_i\}_{i=1}^{\lceil\log n\rceil}$ obtained from the discussion above, there exists a fractional solution to the POISELP connecting the root $r$ and terminals $V_i$ with objective value $6\cdot 2^i$ for all $1\le i\le \lceil\log n\rceil$. 
\end{lemma}

\begin{proof}
Fix $i\in \{1, ..., \lceil\log n \rceil\}$. Consider a vertex $v\in V_i$. Recall that the optimal solution sends one unit of flow from $r$ to $v^s$ and $\sum_j j\cdot f_v(v^jv^s) < 2^i$. Then, by Markov's inequality, at least half of the flow reaches $v^s$ via $v^j$ where $j < 2^{i+1}$. Let us only consider these flows and denote them by $f'_v(e)$ for each $v \in V(G)$ and $e \in E(G')$ (we call this process \textit{filtration}). One can decompose these fractional flows into distinct paths $\mathcal{P}'_v$ in $G'$ with corresponding fractional flows $f'_v$ that sums to at least $1/2$, i.e., $\sum_{P'\in \mathcal{P}'_v}f'_v(P') \geq \frac{1}{2}$. Here we slightly abuse the notation $f'_v$, and use it for both edges and paths in the path decomposition. Note that $f'_v(e)=\sum_{P'\in \mathcal{P}'_v: e \in P'} f'_v(P')$.  Each path $P'\in \mathcal{P}'_v$ naturally corresponds to a path $P \in \mathcal{P}(v, r)$; thus, let $y_v(P)=\sum_{P': P'\in \mathcal{P}'_v\text{ corresponds to } P} f'_v(P')$ for all $P\in \mathcal{P}(v, r)$. Let $x_e=\max_{v\in V_i} \sum_{P: e\in E(P), P\in \mathcal{P}(v, r)} y_v(P)$. 
Also, let $E_e'\subseteq E(G')$ be the set of edges $e'\in G'$ that correspond to $e \in G$. 
Now we check the feasibility of this solution for POISELP. 

For Constraint \eqref{poiselpdeg}, consider some $w\in V(G)$. For each $e\in \delta(w)$ in $G$, there exists $v_e\in V_i$ such that $x_e = \sum_{P: e\in E(P), P\in \mathcal{P}(v_e, r)} y_{v_e}(P)$. 
%Then, let $E_e'\subseteq E(G')$ be the set of edges such that $e'\in E_e'$ belongs in some path $P'\in \mathcal{P}'_{v_e}$. 
%This implies that $x_e = \sum_{e'\in E_e'} f_{v_e}(e')$. Note that $\cup_{e\in \delta_G(w)} E'_e$ is a subset of all the edges in $E_1\subset E(G')$ that is incident to some vertex $w^i$ for $i=\{1, ..., 2^{i+1}\}$.
This implies that
\begin{center}
$x_e =\sum_{P':P'\in \mathcal{P}'_{v_e}, P' \cap E'_e \neq \emptyset} f'_{v_e}(P') \leq \sum_{e'\in E_e'} f'_{v_e}(e')$.
\end{center}
%where $f'_v(e)$ is the flow corresponding to $v \in V(G)$ that passes through $e \in G'$ after filtration (so $f' \leq f$). 
Note that $\cup_{e\in \delta_G(w)} E'_e$ is a subset of all the edges in $E_1\subset E(G')$ that is incident to some vertex $w^j$ for $j=\{1, ..., 2^{i+1}-1\}$.
Then,
\begin{align*}
    \sum_{e\in\delta(w)} x_e & \leq \sum_{e\in \delta(w)} \sum_{e'\in E_e'} f'_{v_e}(e') & \\
    &\leq \sum_{e\in \delta(w)} \sum_{e'\in E_e'} z_{e'} & \text{ by Constraint \eqref{abtlpcons1} and } f'\leq f \\
    & \le \sum_{j=1}^{2^{i+1}-1} \sum_{e'\in \delta(w^j)\cap E_1} z_{e'} & \text{ by the subset observation}\\
    & \le \sum_{j=0}^{2^{i+1}-1} \left( \sum_{e'\in \delta^+(w^j)\cap E_1} z_{e'} + \sum_{e'\in \delta^-(w^{j+1})\cap E_1} z_{e'}\right)\\
    & \le \sum_{j=0}^{2^{i+1}-1} 1 & \text{ by Constraint \eqref{abtlpdeg}}
\end{align*}

Thus, it follows that $\sum_{e\in\delta(w)} x_e\le 2^{i+1}$. For any vertex $v\in V_i$, since the paths $\mathcal{P}'_v$ were chosen to preserve at least half of the flow in $G'$ for vertex $v$, it follows that $\sum_{P\in \mathcal{P}(v, r)} y_v(P) \ge 1/2$. This violates Constraint \eqref{poiselppath} but we will fix this at the end with scaling. For Constraint \eqref{poiselpdiam}, note that since $\mathcal{P}'_v$ is chosen such that the paths enters $v^s$ via $v^j$ where $j\le 2^{i+1}$, their corresponding paths $P$ has length at most $2^{i+1}$. Then, $\sum_{P\in \mathcal{P}(v, r)} len(P)y_v(P) \le 2^{i+1}$ as long as  Constraint \eqref{poiselppath} holds (which we will guarantee in the end). Constraints \eqref{poiselpcons1} and \eqref{poiselpcons2} are satisfied by design. Lastly, simply scale the variables $y_v(P)$ proportionally by $1/\sum_{P\in \mathcal{P}(v, r)} y_v(P)$. Since $\sum_{P\in \mathcal{P}(v, r)} y_v(P) \ge 1/2$, doubling $x_e$ suffices to satisfy Constraint \eqref{poiselpcons1}. Therefore, setting $L_1 = 2^{i+2}, L_2 = 2^{i+1}$ provides a feasible solution with objective value $6\cdot 2^i$. 
\end{proof}

\begin{proof}[Proof of Theorem \ref{th:abtapprox}]
Using the above lemma, for each $V_i$, we can find a solution to the POISELP for $\{r\} \cup V_i$ of value at most $6 \cdot 2^i$. Then we apply Lemma~\ref{lem:latin18} followed by Lemma~\ref{lem:Ravi1994Broadcast} to round this into a schedule that completes in time $O(2^i \frac{\log^2 n}{\log \log n})$, thus incurring a total completion time of $O(|V_i| \cdot 2^i \frac{\log^2 n}{\log \log n})$ for the nodes in $V_i$. 
Repeating this for the sets $V_i$ for $i \in [\log n]$ in that order, the completion time for any node in $V_i$ is $\sum_{j=1}^i O(2^j \frac{\log^2 n}{\log \log n}) = O(2^i \frac{\log^2 n}{\log \log n})$. Summing over all nodes gives a total completion time of $O(\sum_i |V_i| 2^i \frac{\log^2 n}{\log \log n})$. 
%\rrnote{Note that ABTLP is at least half of OPT, loss ude to constraint 1.}
Since the value of ABTLP is at least $\sum_i |V_i| \cdot 2^{i-1}$, we conclude that the ABT schedule has the approximation factor claimed in Theorem~\ref{th:abtapprox}. 
\end{proof}

Note that the method generalized in a straightforward manner for the average multicast time problem where we only require the message to reach a subset of the nodes denoted as terminals. 
%We will highlight an extension of this method to the even more general multicommodity multicast problem (requiring origin-demand pairs to communicate) in Section~\ref{sec:mcmc}. 
\section{AvgRVC on Trees}
\label{sec:tree}
Section~\ref{sec:AvgRVC} showed that any $\alpha$-approximation algorithm for ABT gives an $O(\alpha \log n)$-approximation algorithm for AvgRVC. On one hand, this strategy implies that if one finds a poly-logarithmic approximation to ABT, one also attains a poly-logarithmic approximation for AvgRVC. On the other hand, this strategy would never lead to any approximation better than logarithmic.
Indeed, in the introduction, we showed an example of a class of trees in which there is a $\Theta(\log n)$-factor separation in the optimal values to ABT and AvgRVC. 
Motivated by this, we study AvgRVC on trees, and 
break this logarithmic barrier that is an artefact of going via the ABT solution to obtain a constant-factor approximation for AvgRVC. Thus, the goal of this section is to prove Theorem \ref{treeapprox}. 

%\begin{theorem}
%\label{treeapprox}
%There exists a $40$-approximation for the Average Rooted Vector Clock Latency Problem (AvgRVC) on trees. 
%\end{theorem}

\subsection{Overview of Techniques}

\paragraph{Locally-periodic with regular offsets.} In particular, we show that there exists a near-optimal schedule that is locally-periodic where each vertex receives information from its parent in regular intervals. 
Denote the interval as the period for the vertex.
In optimizing for the AvgRVC problem where the root of the tree is the source of information (with information update period of one), the periods are monotonically increasing on any root to leaf path (since more frequent information update at a node than its parent is wasteful).
Note that even with this notion of period, there is still an offset due to the asynchronicity from when a parent is updated to when it informs the child, which in turn adds to the delay of all the children in its subtree. 
Thus an upper bound on the latency of reaching a node is the sum of these offsets of its ancestors plus the node's own period of hearing from its own parent. 
To simplify the analysis, we only consider such locally-periodic schedules where the offsets at every node (from its parent's update) is the same every time this node is updated by its parent. We call these schedules `regular' and our main result is to show that such a regular schedule is near-optimal.

\paragraph{Designing local periods using LP and DP.} To do this, we first show that there is a locally-periodic schedule for which the sum of the periods at all the nodes is at most 4 times that of an optimal AvgRVC schedule. 
Given this, we design an algorithm to assign periods to nodes in the tree that obeys the telephone constraints and minimizes the total sum of these periods at the nodes and obeys the natural monotonicity requirement on root-to-leaf paths.
%With a slight loss, we can restrict the periods assigned to nodes in the tree to powers of two.
Constructing such an assignment of periods is complicated by the telephone condition where a parent can only inform one child at a given time. This corresponds to the constraint that the inverse of the periods of the children of any node should sum to no more than 1. While a dynamic programming approach is natural for computing such period assignments in a tree, at each node in the tree, the telephone condition is best handled using a matching-based linear program that assigns the right power of two period for each child in its neighborhood going top down. The LP solution at each node can be shown to have at most two fractional variables which can then be rounded and accounted for by scaling.

\paragraph{Designing regular offsets using ABT.} Armed with such a locally-periodic period assignment at the nodes, we next turn this into a regular schedule.
To obtain the offsets, we turn to the optimal ordering of the children at a node in the ABT schedule (according to non-increasing size of the subtrees). 
By using a time averaging argument, we can see that the optimal value of the ABT problem is a lower bound on the optimal average rooted latency.
When calculating the latency contribution of a node, recall that we sum the offsets of its ancestors and its own period. To bound the offsets, note that the offset is at most the period of the parent and hence if the parent's period is halving, the offset is also halving and leads to a geometrically decreasing contribution. Hence the only worrisome offset contributions from ancestors occur when their periods do not reduce (halve) along the path to the root. 
Our last idea is to insist that whenever a pair of ancestors have the same period, the ordering of the children of the higher node must follow the optimal ABT order. By showing how to accomplish this in our locally-periodic schedule, we can bound such worrisome offset contributions to their corresponding terms in the ABT objective which itself is upper bounded by the optimal rooted latency value. 

Put together, our final locally-periodic powers-of-two schedules that also obey the ABT ordering at node neighborhoods where the periods do not change gives near optimal AvgRVC schedules. Details follow.

%Our strategy is to explicitly construct a periodic schedule $\mathcal{S}$ and argue that its average latency at any time is low compared to the optimum. First, we introduce some definitions that help with the interpretation of latency in a periodic schedule on trees, and use them to outline our strategy. 
%For the purpose of this conference submission, we omit the formal proofs for most lemmas and only provide the general intuition. Please refer to the full version \cite{chen2021timeliness} online for more details. 

\subsection{Preliminaries}
Consider an instance of the AvgRVC on a tree $T$ with root $r$. Note that we can view the tree as a directed arborescence since fresh information can only be disseminated downwards from the root towards the leaves. 
\begin{definition}
Given an infinite schedule $\mathcal{S}$ with latency vectors $l^t$, for any vertex $v$, let $t^v_1, t^v_2, ...$ be its \textbf{update sequence}, representing the sequence of time at which $v$ hears a fresh information from its parent. In other words, time $t$ is in the update sequence of $v$ if $l^t_v\le l^{t-1}_v$. Note that without any update at $t$, $l^t_v = l^{t-1}_v +1$.
\end{definition}
 
\begin{definition}
We say a schedule is \textbf{locally periodic} if for every vertex $v$, $t^v_{i+1}-t^v_i=p^v$ for all $i\in \mathbb{N}$ where $\{t^v_i\}_{i\in \mathbb{N}}$ is the update sequence of $v$. We call $p^v$ the \textbf{period of $v$}.
\end{definition}
 Note that in a locally periodic schedule, if $u$ is the parent of $v$, then $p_u\le p_v$ since $v$ cannot hear fresh information more frequently than $u$. Since the root is constantly generating new information, $p^r=1$. 
\begin{definition}
Given a locally periodic schedule, consider a vertex $v$ and its parent $u$. Let $t^v$ be a time in the update sequence of $v$. Let $t^u$ be the largest time in the update sequence of $u$ such that $t^u<t^v$. Then, define $t^v-t^u$ to be the \textbf{offset} of $v$ at time $t^v$. 
\end{definition}
The offset represents the delay of information caused by the asynchronicity of the periods between a vertex and its parent. It follows from definition that the offset of $v$ at any time is at most $p^u$, the period of its parent. 
\begin{definition}
We say a locally periodic schedule is \textbf{regular} if for any vertex $v$, the offset of $v$ at time $t$ is the same for all time $t$ in the update sequence of $v$.
In a regular schedule, let $o^v$ denote the \textbf{offset for vertex $v$}.
\end{definition}
 Note that a locally periodic schedule does not imply a regular schedule. For example, if $v$ has period $4$ and its parent $u$ has a period of $3$, then $v$ can have offsets of $1, 2, 3$ or $4$ at different times in the schedule. 
 Consider a vertex $v$ and its parent $u$. Note that if $p^v$ is not a multiple of $p^u$, then the offset of $v$ varies over time. Thus another condition for a regular schedule is that $p^u|p^v$. 
 
 At a high level, the offsets in a regular schedule is a way of ordering the children of each internal node $u$ in the tree to determine in what order $u$ should inform its children after every time $u$ hears fresh information from the root. 
 We are ready to state an expression upper bounding the average latency in a tree using these definitions. 
\begin{lemma}
\label{lem:reglat}
Consider a regular schedule where vertex $v$ has a period of $p^v$ and an offset of $o^v$. Then, the maximum over time average latency of this schedule is at most $\frac{1}{n}(\sum_{v\in V(G)} \sum_{u\in V(P_v)} o^u +\sum_{v\in V(G)} p^v)$ where $P_v$ is the path from the root to $v$.  
\end{lemma}

\iffalse
\begin{proof}
(OMIT) 
Given $v\in V(G)$, let $P_v= u_0u_1u_2...u_k$ where $u_0=r$ and $u_k=v$. Consider $l^t_v$ the latency of $v$ at time $t$. Since the schedule is regular, for $0\le i<k$, every time $u_i$ hears an information, it takes exactly $o^{u_{i+1}}$ time for $u_i$ to pass the information to $u_{i+1}$. Then, $l^t_v=(\sum_{i=1}^k o^{u_i}) + (t-t')$ where $t'$ is the last time before $t$ at which $v$ hears a fresh information from its parent. Since $v$ hears a new information every $p^v$ time, it follows that $(t-t')\le p^v$ and $l^t_v\le (\sum_{i=1}^k o^{u_i}) +p^v$. Then, our lemma follows. 
\end{proof}

\fi

%The above lemma gives a very concise form to represent the latency of a regular schedule that is independent of time. It also separates the latency into two parts, one dependent on the periods $p^v$ and the second is based on the sum of offsets of the vertices in $P_v$. Refer to the full version \cite{chen2021timeliness} for the proof. Our strategy for proving Theorem~\ref{treeapprox} is to construct a regular schedule and separately bound the two parts of the sum to show that it is comparable to the optimum value. At a high level, given a tree $T$ and a root $r$, we will use dynamic programming to determine the appropriate period to assign to each vertex. Then, we use a solution to the ABT problem on trees to assign the appropriate offsets.

We will show that any optimal schedule can be converted into a locally periodic one where the part of the objective depending on the periods is at most a factor of 4 worse than the average rooted latency (Lemma~\ref{lem:makeperiodic}) in Section~\ref{subsec:period}. This leads us to a solution to the first part of finding a locally periodic schedule where the periods can be bounded (Corollary~\ref{col:bound2}).
Next, we use this periodic schedule and develop a subproblem of making the schedule regular by assigning offsets in Section~\ref{subsec:offset}. In Section~\ref{sec:regular}, we relate the optima of AvgRVC and ABT and use
a solution to the ABT problem in the tree to assign these offsets in constructing a regular schedule whose offsets can be bounded. Section~\ref{subsec:together} puts everything together to prove Theorem~\ref{treeapprox}.

%\iffalse
\subsection{Simple Facts about Sequences of Positive Powers of Two}
We collect here some facts that we will use later.
\begin{claim}
\label{cl:largetail}
    Let $p_1, ..., p_k$ be a non-decreasing sequence of positive powers of $2$ where $\sum_{i=1}^k \frac{1}{p_i} \le 1$. Let $s = 1-\sum_{i=1}^k \frac{1}{p_i}$. Then, the following holds:
    \begin{itemize}
        \item if $s> 0$, then $s \ge \frac{1}{p_k}$ and,
        \item $\frac{1}{p_k} + s \ge \frac{1}{2^k}$. 
    \end{itemize}
\end{claim}

\begin{proof}
    To prove the first statement, since $p_k$ is the largest power of $2$ in the sequence, $p_k/p_i$ is a positive integer for all $1\le i\le k$. Then, it follows that $p_ks = p_k(1-\sum_{i=1}^k 1/p_i)$ is also a non-negative integer. Then $p_ks\ge 1$ whenever $s>0$ as desired.
    
    We prove the second statement by induction. If $k=1$, then $1/p_k +s = 1\ge 1/2$, as required. For $k\ge 2$, let $s' = 1/p_k +s$. By induction, it follows that $1/p_{k-1} + s' \ge 1/2^{k-1}$. By our previous claim, $s'\ge 1/p_{k-1}$. Then, $s'\ge \frac{1}{2}(1/p_{k-1}+s') \ge 1/2^k$, as required. 
\end{proof}
\begin{claim}
\label{cl:partition}
    Let $p_1, ..., p_k$ be a non-decreasing sequence of powers of $2$ where $k\ge 2$ and $\sum_{i=1}^k \frac{1}{p^i}\le 1$. Then, either $\sum_{i=1}^k\frac{1}{p^i} \le \frac{1}{2}$ or there exists $1\le j< k$ such that $\sum_{i=1}^j \frac{1}{p^i}= \frac{1}{2}$.
\end{claim}

\begin{proof}
    Pick the largest index $j$ such that $\sum_{i=1}^j\frac{1}{p^i} \le \frac{1}{2}$. Since $k\ge 2$, $1/p_i\le 1/2$ for all $1\le i\le k$ and thus such $j$ exists. Let $s= 1/2 - \sum_{i=1}^j\frac{1}{p^i} $. If $s=0$, then our claim follows immediately. If $s> 0$, by doubling the values in our sequence, it follows from Claim \ref{cl:largetail} that $2s\ge 2/p_j$. If $j\neq k$, then $1/p_k\le 1/p_j\le s$, contradicting our choice of $j$. Therefore, $j=k$ and $\sum_{i=1}^k\frac{1}{p^i} \le \frac{1}{2}$, proving our claim.
\end{proof}

%\fi

\subsection{Assigning Periods}
\label{subsec:period}

%\iffalse
To bound the sum of periods, we first prove a lemma about general periodic schedules. 

\begin{lemma}
\label{lem:avetime}
    Given a periodic schedule with latency function $l$ and period $P$ and a vertex $v\in V(G)$, let $k$ be the number of times $v$ received fresh information from the root. Then $P^2/k\le 2\sum_{t=P}^{2P-1} l^t_v$ for all $t$. 
\end{lemma}

\begin{proof}
    Let $t_1, ..., t_k$ be the update sequence of $v$ between $P$ and $2P-1$. Since the schedule is periodic, the next time $v$ hears a new information is at $t_{k+1} = t_1+P$. Define $a_i=t_{i+1}-t_i$ for $1\le i\le k$. Essentially, $a_i$ represent the amount of time since $t_i$ $v$ has to wait before hearing new information again. Then, the latency at time $t_{i+1}-1$ is at least $a_i$. Since the schedule is periodic, we have that $\sum_{t=P}^{2P-1} l^t_v = \sum_{t=t_1}^{t_{k+1}-1} l^t_v = \sum_{i=1}^k \sum_{t=t_i}^{t_{i+1}-1} l^t_v \ge \sum_{i=1}^k \sum_{j=1}^{a_i} j \ge \sum_{i=1}^k a_i^2/2$. Since $\sum_{i=1}^k a_i = P$, one can check that $\sum_{i=1}^k a_i^2/2$ is minimized when $a_i=P/k$ for all $1\le i\le k$. Then our lemma follows immediately. 
\end{proof}

%\fi

%The following converts an infinite schedule to a locally-periodic schedule. 

\begin{lemma}
\label{lem:makeperiodic}
    Let $opt_{VC}$ be the optimal value to AvgRVC. There exists a locally-periodic schedule $\mathcal{S'}$ with periods $q^v{}'$ such that $\sum_{v\in V(T)} q^v{}' \le (4n)opt_{VC}$. 
\end{lemma}

The idea is to spread out the occurrence of an edge in the schedule as uniformly as possible and argue that it does not increase the objective function by too much. 
%\rrnote{See the full version \cite{chen2021timeliness} for the entire proof.}

%\iffalse
\begin{proof}
%    (OMIT Lemma 5 and sketch proof of Lemma 6). 
    By Theorem \ref{periopt}, there exists a periodic schedule $\mathcal{S}^*$ whose maximum over time average latency is at most $2opt_{VC}$. Let $P$ be the period of $\mathcal{S}^*$. Let $l^t_v$ be the latency of $v$ at time $t$ in $\mathcal{S}^*$. For every vertex $v\in V(T)$, let $q^v$ be the rational number obtained from dividing $P$ by the number of times $v$ received fresh information from its parent during the time $P$ and $2P-1$. It follows from Lemma \ref{lem:avetime} that $\sum_{v\in V(T)}\sum_{t=P}^{P-1} l^t_v \ge \sum_{v\in V(T)} Pq^v/2$. Note that $\sum_{t=P}^{2P-1} \sum_{v\in V(T)} l^t_v = Pn(opt_{VC})$. Then, it follows that $\sum_{v\in V(T)} q^v \le 2n(opt_{VC})$. 
    
    Since any vertex $v$ cannot receive fresh information more frequently than his parent $u$, $q^u\le q^v$ and $q^v$ satisfies Condition \eqref{regparent}. Since a parent $u$ can only inform one of its children at every time step, it follows that $\sum_{v\in N(u)} P/q^v \le P$. Then $q^v$ also satisfies Condition \eqref{regsum}. Then let $q^v{}'$ be the smallest power of $2$ such that $q^v{}'\ge q^v$. Note that $q^v{}'$ still satisfies Conditions \eqref{regparent} and \eqref{regsum}. 
    Then $q^v{}'$ is a regular assignment of periods. Note that $q^v{}' \le 2q^v$ and thus $\sum_{v\in V(T)} q^v{}' \le 4n(opt_{VC})$, proving our lemma.
\end{proof}

%\fi

Given the above lemma, let us analyze some properties of a locally periodic schedule to determine how to assign periods to nodes in the tree. Consider a vertex $u$ in the tree and let $N(u)$ denote the set of its children. As mentioned before, $p^{v}\ge p^u$ since each child $v\in N(u)$ should naturally wait at least as long as $u$ before it hears new information. Let $P=\Pi_{v\in N(u)} p^{v}$. For each vertex $v\in N(u)$, in a period of $P$ time, $u$ would have informed $v$ exactly $P/p^{v}$ times. However, since at every time step, $u$ can inform at most one of its children, it follows that $\sum_{v\in N(u)} P/p^{v} \le P$. Therefore, the periods of any regular schedule satisfies the following two constraints:
    \begin{align}
        p^u & \le p^v &  \forall u\in V(G), v\in N(u) \label{regparent}\\
        \sum_{v\in N(u)} \frac{1}{p^v}  & \le 1 &  \forall u\in V(G) \label{regsum}
    \end{align}
Consider a vertex $v$ and its parent $u$. Recall that our condition for a locally periodic schedule to be regular is that $p^u|p^v$. We impose a slightly stronger condition that every period must be a power of $2$:
\begin{align}
    p^v & = 2^i \text{ for some } i\in \mathbb{Z}_+ & \forall v\in V(G)
    \label{regpower2}
\end{align}

\begin{definition}
We say $\{p^u\}_{u\in V(G)}$ is a \textbf{regular assignment of periods} if the assignment satisfies Conditions \eqref{regparent}, \eqref{regsum}, and \eqref{regpower2}.
\end{definition} 
Note that a regular assignment of periods is an intermediate subproblem in finding regular schedules that are near-optimal. We will use the periods from this subproblem to eventually assign offsets in Section~\ref{subsec:offset}.
Now, we are ready to formally state the problem of assigning periods to nodes.
\begin{problem}
\label{prob:regperi}
    Given a tree $T$, find a regular assignment of periods to its vertices that minimizes $\sum_{u\in V(G)} p^v$. 
\end{problem}
%We now provide a 2-approximation to the above Problem~\ref{prob:regperi}.

\begin{lemma}
\label{lem:regperi}
Given an instance of Problem \ref{prob:regperi} with optimal value $opt_{RP}$, we can obtain in poly-time a regular assignment of periods where $\sum_{u\in V(G)}p^u\le 2opt_{RP}$.
\end{lemma}

%We first prove a claim about finite sequence of powers of $2$.

The following is a claim about optimal solutions to the Problem \ref{prob:regperi}. 
%\rrnote{Refer to full version \cite{chen2021timeliness} for its proof.}

\begin{claim}
Given an instance of Problem \ref{prob:regperi}, let $q^v$ be an optimal regular assignment of periods. Then $q^v\le 2^{n}$ for all $v\in V(G)$. 
\end{claim}

%\iffalse
\begin{proof}
%(OMIT) 
Suppose for the sake of contradiction that there exists $v$ such that $q^v> 2^n$. Assume without loss of generality that $v$ is a vertex closest to the root that violates this claim. Note that $v\neq r$ since $q^r=1$. Let $u$ be the parent of $v$. We may also assume that $v$ has the largest period among $N(u)$. Consider changing $q^v$ to $q^v{}' = 2^n$ instead. By our choice of $v$, $1/q^u \ge 1/2^n = 1/q^v{}'$. Then, Condition \eqref{regparent} is not violated. Since $q^v{}' = 2^n < q^v$, Condition \eqref{regparent} also holds for any children of $v$. By Claim \ref{cl:largetail}, $1-\sum_{w\in N(u)} \frac{1}{q^w} + \frac{1}{q^v} \ge \frac{1}{2^{|N(u)}} \ge \frac{1}{2^n} = \frac{1}{q^v{}'}$. Then, it follows that $\sum_{w\in N(u)} \frac{1}{q^w}- \frac{1}{q^v} + \frac{1}{q^v{}'} \le 1$ and Condition \eqref{regsum} also remains not violated. Then, the modified assignment remains regular but it improves our objective, a contradiction to the optimality of the original choice.
\end{proof}
%\fi

\begin{proof}[Proof of Lemma \ref{lem:regperi}]
We first use dynamic programming to compute an optimal regular assignment of periods that slightly violates Condition \eqref{regsum}. We make appropriate adjustments at the end to produce the desired approximation.

Given vertex $u\in V(G)$, let $T_u$ denote the  subtree rooted at $u$. Given integer $0\le i\le n$, let $f(u, i)$ denote the minimum value of $\sum_{v\in T_u} p^v$ such that $p^u=2^i$ and the set of assignments $p^v$ obeys Condition \eqref{regparent} and \eqref{regpower2} (but not necessarily \eqref{regsum}). If $u$ is a leaf, then $f(u, i)= 2^i$ for all $0\le i\le n$.

If $u$ is not the root, we use the following auxiliary complete bipartite graph $G'$ to compute $f(u, i)$. Let $N(u)$, the children of $u$ be on one side of the bipartition while the other side contains $n-i+1$ vertices labelled from $i$ to $n$ representing the potential powers of two that can be assigned as the period for these nodes. Every edge $vj$ has a weight $w(vj)$ of $f(v, j)$ and a cost $c(vj)$ of $1/2^j$. Consider the following LP:

\iffalse
\begin{align*}
	\min \qquad & 2^i + \sum_{e\in E(G')}w(e)x_e\notag & \\
\text{s.t.}&\sum_{e\in E(G')}c(e)x_e \le 1;&\ \ 
	\sum_{j=i}^n x_{vj} \ge 1 \ \forall v\in N(u);\ \ 
    x_{uv} \geq 0 \ \forall uv \in E(G') 
\end{align*}
\fi

\begin{align}
	\min \qquad & 2^i + \sum_{e\in E(G')}w(e)x_e\notag &\\
	\text{s.t.} \qquad &\sum_{e\in E(G')}c(e)x_e \le 1 & \label{cons:sum} \\
	& \sum_{j=i}^n x_{vj} \ge 1 & \forall v\in N(u)\notag \\
    & x_{uv} \geq 0 & \forall uv \in E(G') \notag
\end{align}

In this LP, $x_{vj}$ is a binary variable representing if $v$ should have a period of $2^j$ or not. The objective thus represents $f(u, i)$, the sum of all the periods of the vertices in $T_u$. The first constraint is to ensure Condition \eqref{regsum} holds. The second set of constraint are matching constraints, ensuring that every vertex in $N(u)$ gets an assignment of period. The last is simply a non-negativity constraint. 

%Then, using a rank argument, one can show that the optimal solution to the LP has at most two fractional values. Furthermore, the two fractional values corresponds to assigning two different periods to $v$, one particular child of $u$. Then, by forcefully assigning $v$ with the smaller period, one does not increase the objective function and only slightly violates the first constraint. Once the DP assigns the periods to all the vertices, one can fix the violation by doubling the period for every vertex. This only doubles the objective value and doubles the period but achieves a desirable assignment. 
%\rrnote{See \cite{chen2021timeliness} for more details.} 

%\iffalse
%(OMIT replace with above)
Let $x^*$ be an optimal basic feasible solution to the above LP. We claim that there exist at most two fractional value of $x^*$ and all other values are either $0$ or $1$. Note that if there exists $e$ such that $x^*_e> 1$, then we can lower the value to $1$ and produce a feasible solution that does not increase the objective. Since all the $|E(G')|$ non-negativity constraints are independent, by elementary linear algebra, there are at least $|E(G')|$ tight constraints. Since there are in total $1+|N(u)|+|E(G')|$ constraints, it follows that at least $|E(G')|-(|N(u)|+1)$ of the non-negative constraints are tight. In other words, there are at most $|N(u)|+1$ many strictly positive $x^*$'s. Let $F\subseteq E(G')$ be the set of edges whose $x^*$ value is strictly positive (i.e., in the support of $x^*$). 

%Note that if $i<i'$, an LP solution corresponding to $f(u, i')$ is also an LP solution corresponding to $f(u, i)$. Thus, by setting $f(u, i)$ to be the value of the objective function to the associated LP, we ensure that $f(u, i)\le f(u, i')$ for all $i< i'$.  Let $q^v$ be an optimal regular assignment of periods where $\sum_{v\in V(G)} \frac{1}{q^v} = opt$.

To satisfy the matching constraints, it follows that every vertex $v\in N(u)$  is incident to at least one edge in $F$. Then, by the Pigeonhole Principle, there are at most two edges $e, e'\in F$ that are incident to the same vertex $v$ while the rest of the edges form a matching. Thus, in order to satisfy the matching Constraints, it follows that the $x^*$ value is $1$ for  either all edges in $F$ or for all except exactly two edges.

Consider the following rounding procedure. If there exists two edges $vj, vj'$ with fractional $x^*$ value where $j<j'$, then take $x'_{vj}=1, x'_{vj'}=0$. Note that this does not increase the objective function and does not violate the matching constraints. Then it produces an assignment of periods $p^v$ (and subsequently for all its descendants) such that the sum of their periods is at most $f(u, i)$. By construction of $G'$, $p^v\ge p^u$ for all $v\in N(u)$ and thus Condition \eqref{regparent} is satisfied. Condition \eqref{regpower2} is also trivially satisfied. However, the rounding affects Constraint \eqref{cons:sum} and thus might violate Condition \eqref{regsum}. However, since $c(j)=1/2^j\le 1$, this rounding increases the Right-Hand-Side of Constraint \ref{cons:sum} by at most $1$. Thus we are guaranteed that $\sum_{v\in N(u)} 1/p^v \le 2$. 

Now, consider the assignment of periods achieved from rounding the solution associated to $f(r, 0)$. As discussed before, the assignment satisfies Conditions \eqref{regparent} and \eqref{regpower2} but not Condition \eqref{regsum}. However, by simply doubling the period of every vertex in this assignment, we achieve a regular assignment $p^v$ whose sum is at most $2f(r, 0)$.

Consider $q^v$, an optimal assignment of periods to the instance of Problem \ref{prob:regperi} where $opt$ is the optimal value. For any $u\in V(G)$, note that the periods of its children can be easily transformed into a feasible solution to the LP associated to $f(u, i)$ where $p^u=2^i$. Then it follows that $opt\ge f(r, 0)$ and thus our construction provides a $2$-approximation. 
%\fi
\end{proof}

%Putting this all together, we can relate the contribution of the period terms in the constructed schedule to the optimal value of AvgRVC using Lemma~\ref{lem:makeperiodic}.
\begin{corollary}
\label{col:bound2}
    Let $opt_{VC}$ be the optimal value to AvgRVC. We can construct in polynomial time, a schedule $\mathcal{S}$ for which $\sum_{v\in V(T)} p^v \le (8n)opt_{VC}$. 
\end{corollary}

\begin{proof}
  From Lemma \ref{lem:makeperiodic}, there exists a locally periodic schedule with period $q^v{}'$ such that $\sum_{v\in V(T)} q^v{}' \le (4n)opt_{VC}$. Since $q^v{}'$ is a feasible solution to Problem \ref{prob:regperi}, it follows from Lemma \ref{lem:regperi} that for the schedule constructed using the lemma, $\sum_{v\in V(T)} p^v \le 2\sum_{v\in V(T)} q^v{}' \le 8n(opt_{VC})$, proving our corollary. 
\end{proof}

\subsection{Assigning Offsets}
\label{subsec:offset}
We proceed to turn the locally periodic schedule constructed above into a regular schedule where every nodes has a fixed offset. We first make some observation about offsets. Consider a regular schedule $\mathcal{S}$. Let $u$ be a non-leaf vertex. Recall that $o^v\le p^u$ where $v \in N(u)$ by the definition of offset. For $1\le i\le p^u$, let $O^u_i=\{v \in N(u): o^v=i\}$, denoting the set of children of $u$ with offset $i$. Fix $i$, let $P=\Pi_{v
\in O^u_i} p^v$. Within a period of $P$, a vertex $v\in O^u_i$ is informed exactly $P/p^v$ times. However, every $p^u$ steps has at most one vertex with an offset of $i$ and thus $\sum_{v\in O^u_i} P/p^v \le P/p^u$. Then, we have the following constraints:
\begin{align}
    1\le o^v &\le p^u & \forall v\in N(u) \label{offbound} \\
    \sum_{v\in O^u_i} \frac{1}{p^v} & \le \frac{1}{p^u} & \forall u\in V(G), 1\le i\le p^u \label{offsum}
\end{align}
\begin{definition}
Given a regular assignment of periods, we say $o^v$ is a \textbf{regular assignment of offsets with respect to $p$} if Conditions \eqref{offbound} and \eqref{offsum} hold. 
\end{definition}
Note that the conditions for the offsets of $N(u)$ do not depend on the offset of $u$. Thus, given a regular assignment of periods, to obtain a regular assignment of offsets, it suffices to fix a node $u$ and assign the offsets for its children and repeat the process for every possible choice of $u$ independent of the other internal nodes.
%To help with the assignment of offsets, we relax the condition on the given periods. In particular, we relax Condition \eqref{regpower2} to:
%\begin{equation}
%    p^u|p^v \quad \forall v\in N(u) \label{regmultiple}.
%\end{equation}
Recall that the offset is a way of ordering the children. Our goal in the end is to not only come up with a regular schedule with these periods and offsets but also show that the offsets follows some ideal ordering $\pi$. To this end, we prove the following lemma that will be useful in assigning these offsets later.

\begin{lemma}
\label{lem:off}
    Let $u\in V(G)$. Let $p^u$ and $p^v$ for $v\in N(u)$ be a regular assignment of periods. Let $\pi:N(u)\to [|N(u)|]$ be a fixed ordering of $N(u)$. Then, there exists a regular assignment of offsets with respect to $p$ that satisfies the following:
    \begin{equation}
        \text{if } p^v = p^u, \text{ then } o^v\le \pi(v). \label{offorder}
    \end{equation}
\end{lemma}

%The proof itself is not hard but involves several small claims about powers of $2$. We essentially provide an explicit construction of the assignment and show that it satisfies all the necessary properties. See the \cite{chen2021timeliness} for the full proof. 

%Before proving this lemma, we make one more claim about sequences of powers of $2$.
%\iffalse
\begin{proof}
We prove by induction on $|N(u)|+p^u$. If $|N(u)|=1$, then the lemma is trivially true. Assume $|N(u)|\ge 2$

If $\sum_{v\in N(u)} 1/p^v > 1/2$, by Claim \ref{cl:partition}, there exists a bipartition $(A, B)$ such that $A\cup B = N(u)$ and $\sum_{v\in A} 1/p^v = 1/2$. Let $v\in A, w\in B$. Since Claim \ref{cl:partition} achieves the partition by arranging the periods in non-decreasing order, we may assume that $p^v\le p^w$. If $p^v=p^w$, we may assume without loss of generality that $\pi(v)<\pi(w)$.

Let $p^v_A = p^v/2$ for every $v\in A$, $p^w_B=p^w/2$ for every $w\in B$ and let $p^u_A=p^u_B=p^u/2$. Since $\sum_{v\in A} 1/p^v =1/2$ and $\sum_{w\in B} 1/p^w \le 1/2$, it follows that $p_A, p_B$ satisfy Condition \eqref{regsum}. It is also clear that Conditions \eqref{regparent} and \eqref{regpower2} remain true and thus $p_A, p_B$ are regular assignment of periods. Let $\pi_A, \pi_B$ be the ordering of the vertices in $A, B$ respectively that respects the order in $\pi$. In other words, $\pi_A(v)<\pi_A(v')$ if and only if $\pi(v)<\pi(v')$ (similarly for any two vertices in $B$). By induction, there exists assignment of offsets $o_A$ and $o_B$ that is regular and satisfies Condition \eqref{offorder} with respect to $p_A$ and $p_B$ respectively. Then, let $o^v= o^v_A$ for every $v\in A$ and let $o^w= o^w_B+p^u/2$ for every $w\in B$. We now prove that $o$ is the desired offset assignment. 

Let $v\in A, w\in B$. Note that $1\le o^v=o^v_A\le p^u_A=p^v/2 \le p^u$ and $1\le o^w=o^w_B+p^u/2 \le p^u/2 +p^u/2 = p^u$. Then, the offset assignment $o$ satisfies Condition \eqref{offbound} with respect to $p$. If $1\le i\le p^u/2$, then $\sum_{v\in O^u_i} 1/p^v = \sum_{v: o^v_A = i} 2/p^v_A \le 2/p^u_A = p^u$. If $p^u/2 < i\le p^u$, then $\sum_{w\in O^u_i} 1/p^w = \sum_{w: o^v_B = i-p^u/2} 2/p^w_B\le 2/p^u_B = p^u$. Thus assignment $o$ satisfies Condition \eqref{offsum} and is regular with respect to $p$.

To check Condition \eqref{offorder}, suppose $v\in A$ and $p^v=p^u$. Then, $o^v=o^v_A\le \pi_A(v) \le \pi(v)$. If $w\in B$ and $p^w=p^u$, then for every $v\in A$, $p^v\ge p^w = p^u$. It follows from Condition \eqref{regparent} that $p^v=p^u$ for every $v\in A$. By Claim \ref{cl:partition}, $\sum_{v\in A} 1/p^v = 1/2$. This implies that $|A|= p^u/2$. Since $\pi_B$ respects the order of $\pi$ and all vertices in $A$ appears before $w$ in $\pi$, $\pi_B(w) \le \pi(w)-|A|$. Then, $o^w = o^w_B + p^u/2 \le \pi_B(w) +|A| \le \pi(w)$. Thus, the offset assignment $o$ satisfies the condition \eqref{offorder} with respect to $p$, as required.

Lastly, if $\sum_{v\in N(u)} 1/p^v \le 1/2$, let $p^v{}'=p^v/2$ and $p^u{}'=p^u/2$. Note that $p^v{}'$ is also a regular assignment of periods. Then, by induction, there exists an assignment of offsets $o^v$ that is regular with respect to $p^v{}'$. With similar arguments as before, it is easy to check that $o$ is also regular and satisfies Condition \eqref{offorder} with respect to $p$, proving our lemma.

%let $S=\{v: p^v=p^u\}$, represent the set of children with the same period as $u$. It follows from Condition \eqref{regsum} that $|S|\le p^u$. Let $c=p^u-|S|$. For every $v\in S$, assign $o^v\in [|S|]$ based on their order in $\pi$. If $s>0$, by doubling the numbers in our sequence, it follows from Claim \ref{cl:largetail} that $2s \ge 2/p^j$. However, since $p^j\ge p^{j+1}$, it follows from our choice of $j$ that $j=k$. 

%Let $N'(u)=N(u)\setminus S$. If $N'(u)\neq \emptyset$, for every $v\in N'(u)$, let $p^v{}'= cp^v/p^u$. Let $p^u'=c$. Since $p^u|p^v$, it follows that $p^u'|p^v{}'$ and thus Conditions \ref{regparent} and \ref{regmultiple} holds. Note that $\sum_{v\notin S} \frac{1}{p^v} \le 1-\sum_{v\in S} \frac{1}{p^v} = 1-\frac{|S|}{p^u} = \frac{c}{p^u}$. Then, $\sum_{v\in N'(u)} \frac{1}{p^v{}'} = \sum_{v\notin S} \frac{p^u}{cp^v} \le \frac{p^u}{c} \frac{c}{p^u} = 1$. Therefore $p^v{}'$ satisfies Condition \eqref{regsum} as well. It is clear from definition that $|N'(u)|\le |N(u)|$ and $p^u'< p^u$. Furthermore, if $|S|> 0$ then $|N'(u)|< |N(u)|$. 
\end{proof}

%\fi

\subsection{Constructing a Regular Schedule}
\label{sec:regular}
Let $p^v$ be a regular assignment of periods that is a $2$-approximation to Problem \ref{prob:regperi} obtained from Lemma \ref{lem:regperi}. To assign the offsets, we need to construct an order $\pi$ of the children of any node $u$ in order to apply Lemma \ref{lem:off}.
%For this, we first characterize the order $\pi$ that gives an optimal solution to ABT on a tree, and then use this order for setting offsets below.
\begin{lemma}
\label{lem:ABTopt}
    Given a tree $T$ with root $r$, let $T_v$ be the subtree rooted at $v$ for all $v\in V(T)$. The optimal schedule for ABT on $T$ is, for every vertex $u$, inform its children $N(u)$ in non-increasing order based on the size of $T_v$ for $v\in N(u)$. 
\end{lemma}

\begin{proof}
    Suppose for the sake of contradiction that in an optimal schedule, there exists a vertex $u$ that informs $v$ before $w$ but $|T_v|<|T_w|$. Let $t, t'$ be the time at which $v, w$ respectively hears the information from $u$ for the first time. Consider the schedule where $u$ informs $w$ at $t$ and informs $v$ at $t'$ instead. This swap causes every vertex in $T_w$ to receive information sooner by exactly $t'-t$ time steps while every vertex in $T_v$ is delayed by $t'-t$. Since $|T_w|>|T_v|$, this reduces the total delay, a contradiction. 
\end{proof}

\begin{lemma}
\label{lem:VCABT}
    Given a tree $T$ and a root $r$, let $opt_{VC}$ and $opt_{ABT}$ be the optimal value to AvgRVC and ABT respectively. Then, $2opt_{VC}\ge opt_{ABT}$. 
\end{lemma}

\begin{proof}
    Let $\mathcal{R}$ be a schedule that is a witness to $opt_{VC}$. Let $t$ be a time such that the average latency at $t$ is $opt_{VC}$. This implies for every vertex $v\in V(T)$, the path from $r$ to $v$ appears in sequential order in $\mathcal{R}$ from time $t-l^t_v$ to $t$. Let $L$ be the largest latency value of a node at time $t$. Let $\mathcal{R}_i$ be the subschedule of $\mathcal{R}$ from time $t-2^i$ to $t$ for $0\le i\le I =\lceil \log L \rceil$. Let schedule $\mathcal{R}'$ be the concatenation of the subsechdules $\mathcal{R}_0,..., \mathcal{R}_I$. We claim that in schedule $\mathcal{R}'$, every vertex $v$ hears an information from the root before time $2l^t_v$. 
    
    Consider a vertex $v$. Let $0\le i\le I$ such that $\lceil \log l^t_v \rceil = i$. Then, in schedule $\mathcal{R}_i$, $v$ receives information from $r$. Since $|\mathcal{R}_j| =2^j$ for all $0\le j\le I$, it follows that the sequence in $\mathcal{R}_i$ can be completed in $\mathcal{R}'$ by time $2^{i+1}$. Since $2^{i+1}\le 2 l^t_v$, our lemma follows immediately. 
\end{proof}

Given vertex $u\in V(T)$ and $v\in N(u)$, let $T_v$ be the subtree rooted at $v$. Let $\pi^u:N(u)\to [|N(u)|]$ be an ordering of the children of $u$ in non-increasing order based on $|V(T_v)|$. 
The above two lemmas justify the reason we choose this particular $\pi_v$.  
Using this ordering $\pi^u$ and periods $p^u, \{p^v\}_{v\in N(u)}$, we apply Lemma \ref{lem:off} to obtain a regular assignment of offsets $o^v$. Then, we prove the following:
\begin{lemma}
\label{lem:regsched}
There exists a regular schedule $\mathcal{S}$ such that each vertex $v$ has a period of $2p^v$ and an offset of at most $2o^v$. 
\end{lemma}

The idea is to construct the schedule top-bottom starting with the root. We show that at each node, if the periods and offsets of its children are known and obey the ``regular'' constraints, then a regular schedule can be built. 
%\rrnote{See the full version \cite{chen2021timeliness} for the proof in its entirety.} 

%\iffalse 
%(SKETCH Proof of Lemma 11 and omit Lemma 12) We first prove a claim about scheduling regular periods.

\begin{lemma}
\label{lem:repeat}
Let $V$ be a set of vertices. For each $v\in V$, let $p^v$ be a power of $2$ such that $\sum_{v\in V} \frac{1}{p^v} = 1$. Then, there exists a regular sequence of the vertices in $V$ such that the period of vertex $v$ is $p^v$. 
\end{lemma}

Here, a regular sequence means an infinite sequence of the vertices such that for any vertex, the time difference between any two consecutive occurrence of that vertex in the sequence is the same. 

\begin{proof}
    We induce on $|V|$. If $|V|=1$, then the schedule is simply repeating that vertex. For $|V|\ge 2$, by Claim \ref{cl:partition}, $V$ can be partitioned into two parts $V_1, V_2$ such that $\sum_{v\in V_i} 1/p^v =1/2$ for $i=1, 2$. It follows that $p^v\ge 2$ for all $v\in V$. Let $q^v=p^v/2$. By induction, there exists a regular schedule $\mathcal{S}_i$ for $V_i$ such that the period of any vertex is $q^v$, for $i=1, 2$. Consider interleaving the two schedules by alternating among them. Note that the combined schedule is also regular and every vertex $v$ has period $2q^v=p^v$, as required. 
\end{proof}

\begin{proof}[Proof of Lemma \ref{lem:regsched}]
    We first build a schedule that ignores the matching constraint in the telephone model where we allow a vertex $u$ is allowed to send and receive information in the same time step. However, $u$ is still only allowed to send information to one of its children in a single time step. Finally, we use the idea of alternating odd and even matchings to produce a proper schedule from the path scheduled at each step. 
    
    Given vertex $u$, let $O^u_i\subseteq N(u)$ denote the subset of children whose offset value is $i$, for $1\le i\le p^u$. Let $\{t^v_j\}_{j\ge 0}$ be the update sequence of $v$. We build our schedule by inductively defining the update sequence for every vertex starting from the root. At every vertex $v$, we show that $v$ receives a new information according to its update sequence and, $v$ has period $p^v$ and offset $o^v$.
    
    At the root $r$, $t^r_i=i$ for all $i\ge 0$. It is clear that $r$ is updated according to its update sequence and has the correct period and offset. Now, assume $u$ received its first fresh information at $t^u_0$ and hears a new information every $p^u$ steps. Let us define the update sequence for its children.
    
    Fix $1\le i\le p^u$ and consider $O^u_i$. Since $o$ is a regular assignment of offset, Condition \eqref{offsum} holds. Let $q^v = p^v/p^u$. It follows from Condition \eqref{offsum} that $\sum_{v\in O^u_i} q^v\le 1$. Let $2^k$ be the largest value of $q^v$. Obtain $\tilde{O}^v_i$ by adding enough dummy vertices into $O^u_i$ with period $2^k$ such that $\sum_{v\in \tilde{O}^u_i} 1/q^v =1$. Let $\mathcal{\tilde{S}}$ be the regular sequence of $\tilde{O}^u_i$ obtained by applying Lemma \ref{lem:repeat}. For $v\in O^v_i$, let $\{\tilde{t}_j^v\}_{j\ge 0}$ be the sequence of time where $v$ appears in schedule $\mathcal{\tilde{S}}$. By Lemma \ref{lem:repeat}, $\tilde{t}_{j+1}^v-\tilde{t}_j^v = q^v$ for all $j\ge 0$. Then, let $t_j^v = t_0^u+i+p^u \tilde{t}^v_j$. We now show that this corresponds to the situation where every time $u$ hears a new information, $u$ waits for $i$ additional steps and informs the vertices in $O^u_i$ according to the sequence $\mathcal{\tilde{S}}$.
    
    First, note that $t^v_i$ is regular and has interval $p^uq^v = p^v$. Since $u$ hears an update every $p^u$ steps, it follows that $u$ receives a new information at $t^u_0+p^u \tilde{t}^v_j$ for all $j\ge 0$. Then, it follows that $v$ has an offset of $t^v_j-(t^u_0+p^u \tilde{t}^v_j) = i$. Lastly, it remains to show that at every time step, $u$ only sends information to one of its children.
    
    Suppose for the sake of contradiction that at some time $t$, according to the above update sequence, $v$ and $w$ receives information from $u$ at the same time. This implies, $v$ and $w$ has the same offset $i$. Let $t=t^u_0+i +p^u \tilde{t}$. Then, $v$ and $w$ both appear at time $\tilde{t}$ in sequence $\mathcal{\tilde{S}}$, a contradiction. 
    
    Then, the update sequences induces a regular schedule where every vertex $v$ has period $p^v$ and offset $o^i$ and no vertex $u$ sends information to two of its children at the same time. This implies that at every time step, the edges used to send information induces a collection of paths. For every time step, consider first sending information along the edges in the paths on the odd levels, then send information along the edges in the paths on the even levels. This ensures no vertex receives and sends information in the same time step, satisfying the telephone model. It doubles the period of any vertex and worsens the offset by at most a factor of $2$, as required. 
    
\end{proof}

%\fi

\subsection{Putting it Together: Upper Bounding the Constructed Schedule}
\label{subsec:together}
In this subsection, we bound the average latency of the schedule $\mathcal{S}$ constructed from the previous sections.
Recall that $p^v$ is a regular assignment of periods that is a $2$-approximation to Problem \ref{prob:regperi} obtained from Lemma \ref{lem:regperi}.
Then using the ordering $\pi$ from Lemma ~\ref{lem:ABTopt} in Lemma~\ref{lem:off}, we construct a regular schedule $\mathcal{S}$. 
It follows from Lemma \ref{lem:reglat}, Corollary~\ref{col:bound2} and Lemma \ref{lem:regsched} that $\mathcal{S}$ has an average latency of at most $\frac{1}{n}(2\sum_{v\in V(G)} \sum_{u\in P_v} o^u + 2\sum_{v\in V(G)} p^v)$. 

We will bound the two sums separately. 
Note that to bound the periods, we can use Corollary~\ref{col:bound2}. 
We bound the offset terms below in terms of $opt_{VC}$ and the sum of the periods.
In the following lemma, the offset terms along a path to the root are broken into chains where the offset values decrease by a factor of at least two, and those where the offset terms are the same as the parent. The former are charged to the period terms via a geometric sum, while the latter are charged to a copy of $opt_{VC}$ by going via $opt_{ABT}$.
\begin{lemma}
\label{lem:bound1}
    Let $opt_{VC}$ be the optimal value to AvgRVC. For every vertex $v$ let $P_v$ be the path from the root to $v$. For the schedule $\mathcal{S}$ constructed as explained above, $\sum_{v\in V(T)} \sum_{u\in P_v} o^u \le (2n)opt_{VC} +\sum_{v\in V(T)} p^v$. 
\end{lemma}

\begin{proof}
    Let $V_0\subseteq (V(T)\setminus \{r\})$ be the set of vertices $v$ whose period $p^v$ is the same as its parent's period $p^u$. Consider a vertex $v\notin V_0$. Let $u$ be the parent of $v$. Since $p^u\le p^v$ (using Condition \eqref{regparent}) and the periods are all powers of $2$ (Condition \eqref{regpower2}), $p^v\ge 2p^u$. By definition of offset, $o^v\le p^u$ and thus $o^v\le p^v/2$.
    
    Then, fix a vertex $w\in V(T)$. It follows that $\sum_{u\in (V(P_w)\backslash V_0)} o^u \le \sum_{u\in (V(P_w)\setminus V_0)} p^u/2$. Since we only sum over vertices not in $V_0$, every subsequent term in the sum is at least twice as large as the previous term. Overall, the whole sum is at most as large as twice the last term. Therefore, the sum is at most $p^w$. 
    
    Recall that from Lemma \ref{lem:ABTopt}, $\sum_{v\in V(P_w)}\pi(v)$ is the delay of $w$ in an optimal ABT schedule. Note that for a vertex $v\in V(P_w)\cap V_0$, it follows from Condition~\eqref{offorder} in Lemma \ref{lem:off} that $o^v\le \pi(v)$. Then $$\sum_{w\in V(T)}\sum_{v\in V(P_w)\cap V_0} o^v \le \sum_{w\in V(T)}\sum_{v\in V(P_w)\cap V_0} \pi(v) \le \sum_{w\in V(T)} \sum_{v\in V(P_w)} \pi(v) = n(opt_{ABT})$$ where $opt_{ABT}$ is the value of the optimal ABT solution. Then, by Lemma \ref{lem:VCABT}, \\$\frac{1}{n}\sum_{w\in V(T)}\sum_{v\in V(P_w)\cap V_0} o^v \le 2opt_{VC}$. 
    Therefore, 
    $$\sum_{w\in V(T)}\sum_{v\in V(P_w)} o^v = \sum_{w\in V(T)}\sum_{v\in V(P_w)\cap V_0} o^v + \sum_{w\in V(T)}\sum_{v\in V(P_w)\setminus V_0} o^v \le (2n)opt_{VC}+\sum_{w\in V(T)} p^w$$ proving our lemma. 
\end{proof}

\begin{proof}[Proof of Theorem \ref{treeapprox}]
    Given tree $T$, construct schedule $\mathcal{S}$ by applying Lemma \ref{lem:regsched}. It follows from Lemma \ref{lem:reglat} and Lemma \ref{lem:regsched} that the average latency of $\mathcal{S}$ at any time is bounded by $\frac{1}{n}(\sum_{v\in V(T)}\sum_{u\in P_v} 2o^u + \sum_{v\in V(T)} 2p^v)$. Then, by Lemma \ref{lem:bound1} and Corollary \ref{col:bound2}, our result follows immediately.
\end{proof}

%\newpage
%\input{mcmc}
%\input{vcbcgap-focs}
%\input{maxtree}
\section{Open Problems}
We have introduced the average broadcast time problem in general graphs and provided the first poly-logarithmic approximation algorithm for it, by reducing it to a series of maximum time versions using an LP relaxation solution.
The natural open questions arising from our work are to improve the approximation ratios for the algorithms perhaps even by considering special cases such as planar graph instances.
%Basic questions about the problem remain: Can this version be approximated better than the maximum version for any of the different versions (broadcast/multicast/multi-commodity in trees/planar/general graphs) of problems we considered?

We related the vector clock latency problems to their delay time versions to derive many of our results. In the process, the subgraphs used for the information transmissions have mainly been restricted to appropriate spanning trees which are sufficient for the Max and Average broadcast time problems. Is there an advantage to using edges other than in a spanning tree for the vector clock problems, or are such more complex solutions only within a constant factor of the former? E.g., in a simple cycle, repeatedly choosing the alternate matchings is the best solution for the vector clock latency problems, but restricting oneself to a tree and the two matchings in the tree still provides a solution that is only a constant factor away for the average and the maximum latency problems. However in a generalized multicommodity version~\cite{latin18} where only the latencies between the endpoints of the the edges of the cycle are minimized, the maximum latency version of the problem suffers a very large loss when restricted to a path. What is the boundary of communication requirements that necessitate looking at more than trees in devising latency schedules?

%Finally, our reductions highlight the multi-commodity multicasting problem of minimizing the maximum delay as the only one with no poly-logarithmic approximations yet even for the delay version. 
%The missing entries in Table~\ref{results} also motivate the design of LP-rounding algorithms for the minimum multi-commodity multicasting problem as well as for the related generalized poise problem.
%involving the source-destination pairs in these graphs of finding a (potentially disconnected) subgraph minimizing the sum of the maximum degree and the maximum source-destination distance in the subgraph.
%\rrnote{Mention LP rounding of generalized poise LP implies similar guarantees for average time and latency problems.}
\iffalse
\begin{enumerate}
    \item multicommodity mutlicast LP rounding
    \item l-p norms of rooted latencies
    \item rooted weighted latencies
    \item radio and wireless
\end{enumerate}
\fi
%\input{aidin-rel}

\bibliographystyle{abbrv}
\bibliography{biblio-2022}
\begin{appendix}
\section{Missing Proofs from Section~\ref{sec:prelim}}
\label{sec:app}

\begin{proof}[Proof of Theorem \ref{periopt}]
	The new schedule is simply a repetition of the first $opt+1$ steps of the original schedule $\mathcal{S}$. Note that $l^t{}'=l^t$ for $0\le t\le opt$ since the two schedules start off with the same sequence of matchings. 
	
	First we show that after one period, the sequence of latency vectors $\{l^t{}'\}_{t\ge opt}$ becomes periodic. 
	Fixing a vertex $v$, consider the sequence of latencies $\{l^t_v{}'\}_{t\ge 0}$. In the first $opt+1$ steps, $v$ would have received some information from $r$, otherwise $l^{opt}_v= opt+1$ and $||l^{opt}_v{}'||\ge opt+1$, a contradiction. Then, let $t'\le opt$ be the first time $v$ hears something about $r$ and let $l^{t'}_v{}'=\lambda'$. In other words, from time $t'-\lambda'$ to time $t'$, the schedule managed to send information from $r$ to $v$. Since our new scheme is periodic, this implies that for every $k\in \mathbb{N}$, the scheme is able to send information at time $t'-\lambda'+k(opt+1)$ from $r$ and reaches $v$ at time $t'+k(opt+1)$. Thus the latency at time $t'+k(opt+1)$ is guaranteed to be $\lambda'$. Since the schedule is periodic, it follows that the sequence $l^t_v{}'$ stays periodic after $t'$ and thus stays periodic after time $opt$. Since the above is true for any vertex $v$, the latency vector $l^t{}'$ also stays periodic after time $opt$.
	
	Knowing that the latency vector is periodic after one period, we only need to consider the vectors of the first two periods. In the first period, since the schedule is the same as before, it is clear that the norm at any time is still upperbounded by $opt$. Now consider a time $opt+k$ for some $0\le k< opt+1$ and the latency of information from $r$ to $v$. At time $opt$, the latency at $v$ is $l^{opt}_v$. Since the schedule is periodic, any further delay from time $opt$ to time $opt+k$ is at most the delay caused from time $0$ to time $k$. Note that the delay from time $0$ to time $k$ is exactly the latency at time $k$. In other words, $l^{opt+k}{}'\le l^{opt}{}'+l^{k}{}'$. Then, it follows that $||l^{opt+k}{}'||\le ||l^{opt}{}'+l^{k}{}'||\le ||l^{opt}{}'||+||l^{k}{}'||\le 2opt$, as required. 
\end{proof}

\begin{proof}[Proof of Theorem \ref{redn}]
%    (OMIT)
    Given an instance of MaxRVC with graph $G$ and root $r$, let $opt_{VC}$ be its optimal value. Consider the instance of the Minimum Broadcast-Time Problem on $G$ with root $r$ and let $opt_{BC}$ be the optimal value. Consider a schedule with latency vector $l^t$ such that $||l^t||_{\infty} = opt_{VC}$. Then, it follows that the first time any vertex $v$ hears an information from the root is at most $opt_{VC}$. Therefore, this schedule has a minimum broadcast-time of at most $opt_{VC}$ and thus $opt_{VC}\ge opt_{BC}$. 
    
    Let $\mathcal{S}$ be a schedule that is an $\alpha$-approximation to the instance of Minimum Broadcast-Time. In other words, this schedule guarantees all vertices receives an information from $r$ before time $\alpha opt_{BC}$. Without loss of generality, we may assume that $\mathcal{S}$ has length at most $\alpha opt_{BC}$ since every vertex would have been informed by then. Consider the schedule $\mathcal{S}'$ where we simply repeat schedule $\mathcal{S}$ ad infinitum. Let $l^t{}'$ be the latency vector for $\mathcal{S'}$. Fix a vertex $v$ and let $t'$ be the first time $v$ hears an information from the root $r$. Note that the latency at $t'$ is at most $\alpha opt_{BC}$. Since the schedule is periodic, $v$ is guaranteed to hear something from $r$ at time $t'+T$ where $T$ is the length of the schedule $\mathcal{S}$. 
    %Furthermore, every time $v$ hears an information, its latency is dropped to $l^t{}'_v$. 
    Since $T \leq \alpha opt_{BC}$, it follows that the latency of $v$ is at most $2(\alpha opt_{BC})$. Our theorem follows immediately by combining the two inequalities relating $opt_{C}$ and $opt_{BC}$.

\end{proof}

\section{A Logarithmic Gap Between the Optima of AvgRVC and ABT in Trees}
\label{sec:gap}
Theorem \ref{th:sandwich} allows us to obtain a periodic Vector Clock schedule whose latency is at most an $O(\log n)$-factor more than the optimal ABT value. In this section, we show that this factor is unavoidable even in trees.

\begin{theorem}
\label{AvgGap}
    For every $n \in \mathbb{N}$ there exists a tree $T$ such that the optimal objective value of AvgRVC on $T$ is at least $\Omega(\log^2 n)$. Meanwhile, the optimal objective value of ABT on $T$ is at most $O(\log n)$. 
\end{theorem}

%(MOVE Theorem claim and proof sketch to intro along with description of tree below - OMIT sections 5 and 6).
Consider the following tree $T$. Let $r$ be the root of $T$. Let $v_1, ..., v_d$ be the children of $r$. Suppose each child $v_i$ is the root of a complete binary tree of size $n_i=\lceil n/(2i^2)\rceil$. Choose $d$ such that $\sum_{i=1}^d n_i = n$. Note that $d\ge \sqrt{n}$.
The following lemma provides a lower bound to the optimal value of AvgRVC on $T$. 

\begin{lemma}
\label{lem:gap1}
    Let $opt_{RVC}$ be the optimal value of AvgRVC rooted at $r$ in $T$. Then $opt_{RVC}\ge \frac{1}{16}\log^2 n$
\end{lemma}

\begin{proof}
It follows from Lemma \ref{lem:makeperiodic} that there exists a locally periodic schedule with periods $q^v$ such that $\sum_{v\in V(T)} q^v \le (4n)opt_{RVC}$. For ease of notation, let $q_i$ be the period of vertex $v_i$ for $1\le i\le d$. Since any descendent of $v_i$ has period at least as large as $q_i$, it follows that $(4n)opt_{RVC}\ge \sum_{i=1}^d n_iq_i$. Since $\sum_{i=1}^d 1/q_i\le 1$, it is easy to check that $\sum_{i=1}^d n_iq_i$ is minimized when $q_i=\lambda/ \sqrt{n_i}$ where $\lambda = \sum_{i=1}^d \sqrt{n_i}$. This gives a minimum value for $\sum_{i=1}^d n_iq_i$ of $\lambda^2$. Note that $\lambda = \sum_{i=1}^d \sqrt{n_i} \ge \sum_{i=1}^d \frac{\sqrt{n}}{i} \ge \sqrt{n} \log d$. Since $d\ge \sqrt{n}$, $\lambda^2\ge n\log^2 n/4$, proving our lemma. 
\end{proof}

Now, to upper bound the broadcasting problem we claim the following:

\begin{lemma}
\label{lem:gap2}
    There exists a broadcasting schedule for $T$ such that the objective value for ABT is at most $4.5\log n$. 
\end{lemma}

\begin{proof}
Consider the schedule where we always inform the children of $r$ in increasing order according to index $i$. Since all subtrees are perfect binary trees, all other nodes simply inform their children as soon as they receive information, in any order they choose. For a child $v^i$ of $r$, it is informed at time $i$. Every descendent of $v^i$ can be informed within an additional time of $2\log n_i$. Then, the sum of delays is at most $\sum_{i=1}^d n_i(i+2\log n_i)$. 

We break the sum into two parts. 
If $1\le i\le \log n$, then, $\sum_{i=1}^{\log n} n_i(i+2\log n_i) \le \sum_{i=1}^{\log n} n_i(3\log n)\le (3n)\log n$. If $\log n\le i \le d$, then $\sum_{i=\log n}^d n_i(i+2\log n_i) \le \sum_{i=\log n}^d n_i(3i) \le \sum_{i=\log n}^d 1.5n/i \le 1.5n \log d \le 1.5n \log n$. Thus, $\sum_{i=1}^d n_i(i+2\log n_i) \le (4.5n)\log n$ and the lemma follows immediately. 
\end{proof}

\begin{proof}[Proof of Theorem \ref{AvgGap}]
The theorem follows from Lemmas \ref{lem:gap1} and \ref{lem:gap2}. 
\end{proof}
%\input{appendix}
%\end{appendix}
%\input{rooted}
\end{appendix}
\end{document}